\documentclass[authoryear,5p,times,twocolumn]{elsarticle}
\usepackage[T1]{fontenc}
\usepackage[latin9]{inputenc}
\usepackage{amsthm}
\usepackage{amsmath}
\usepackage{amssymb}
\usepackage{thmtools, thm-restate}
\usepackage[hidelinks]{hyperref}
\usepackage{graphicx}
\usepackage{placeins}
\usepackage{epstopdf}
\usepackage{units}
\usepackage{mathrsfs}
\usepackage{algorithm}
\usepackage{etoolbox}
\AtBeginEnvironment{algorithm}{\small}
\usepackage{relsize}
\usepackage[scr=rsfs,scrscaled=1]{mathalfa}
\usepackage{xcolor}
\usepackage[capitalise,nameinlink,noabbrev]{cleveref}
\allowdisplaybreaks[3]

\usepackage[compact]{titlesec}
\titleformat{\paragraph}[runin]{\normalfont\bfseries}{\theparagraph}{1em}{}

\crefname{equation}{}{}
\crefrangelabelformat{equation}{(#3#1#4)-(#5#2#6)}
\crefname{assumption}{Assumption}{Assumptions}

\newtheorem{lemma}{Lemma}
\newtheorem{proposition}{Proposition}
\newtheorem{corollary}{Corollary}
\newtheorem{remark}{Remark}

\DeclareMathAlphabet\mathbfcal{OMS}{cmsy}{b}{n}

\def\T{\mathsf{T}}
\newcommand*{\E}[1]{\mathsf{E} \left\{ #1 \right\}}

\def \0{\mathbf{0}}
\DeclareMathOperator{\col}{col}
\DeclareMathOperator{\blkdiag}{blkdiag}

\DeclareMathOperator{\trace}{tr}

\newtheorem{theorem}{Theorem}

\makeatletter

\theoremstyle{plain}

\theoremstyle{remark}

\newtheorem{assumption}{Assumption}
\newtheorem{definition}{Definition}
\makeatother

\providecommand{\remarkname}{Remark}
\providecommand{\theoremname}{Theorem}
\graphicspath{{./}{Figures/}}

\let\cite\citep
\begin{document}

\begin{frontmatter}
\title{Multi-Rate Distributed Unscented Kalman Filtering Under Collective Detectability}
\author[bv,ak]{Mohammad Ali Abooshahab\corref{cor1}\fnref{fn1}}
\ead{mohammad.abooshahab@bouvet.no}
\author[ntnu]{Morten Hovd}
\ead{morten.hovd@ntnu.no}
\cortext[cor1]{Corresponding author.}
\fntext[fn1]{This work was performed while M.\,A. Abooshahab was with the Department of Engineering Cybernetics, NTNU, Trondheim, Norway.}
\affiliation[bv]{organization={Bouvet}, country={Norway}}
\affiliation[ak]{organization={Arta Kyb}, city={Stavanger}, country={Norway}}
\affiliation[ntnu]{organization={Department of Engineering Cybernetics, Norwegian University of Science and Technology (NTNU)}, city={Trondheim}, country={Norway}}
\begin{abstract}
We develop a multi-rate distributed unscented Kalman filter for nonlinear sensor networks without a fusion center. Local continuous--discrete UKFs run on a common base grid, update at their own sampling instants, and exchange estimates within one-hop neighbourhoods. Under a windowed collective-detectability condition, diffusion of information matrices gives uniformly bounded quotient covariances although no individual node need be observable. When the collectively invisible subspace is trivial and the statistical-linearization discrepancies satisfy an explicit coherence bound, the information-weighted fused mean yields exponentially bounded mean-square errors without a contraction--mixing condition. This contrasts with the arithmetic diffusion mean, which requires one. A distributed $H_{\infty}$ variant uses a fixed attenuation penalty. A finite-prefix feasibility check and the one-hop detectability Gramian yield a uniform regularized-information margin; an explicit local-contraction--mixing condition then yields mean-square error boundedness. The developed methods are tested on a stochastic nonlinear benchmark and on a three-inertia network in which every node misses at least one mode. Information-mean diffusion is more accurate than five-round covariance-averaging consensus at one fifth of its communication. It remains stable under $3\!:\!2$ multi-rate sampling, where the covariance-averaging consensus and the diffusion mean both diverge, and it keeps the quotient covariances bounded.
\end{abstract}
\begin{keyword}
Kalman filtering \sep distributed filtering \sep sensor networks \sep sensor fusion \sep multi-rate systems \sep dynamic state estimation
\end{keyword}
\end{frontmatter}
\section{Introduction}

Cheap sensors with built-in communication have made networked monitoring practical for large-scale systems~\cite{CooperativeControl,SayedBook}, and with it the need for Kalman-type filters that run on the sensor network itself, without a fusion center~\cite{cattivelli2010diffusion,RezaSN,CooperativeControl}. However, initial approaches to Kalman filtering over sensor networks either used a central processing unit to fuse the observations or relied on all-to-all communication so that every sensor could replicate that unit's operations~\cite{FisrtDisCont,FirstDisKalman}. A single processing unit is a single point of failure and demands elaborate communication and timing protocols, while all-to-all exchange floods the network with traffic; neither scales to networks the size of a power grid.

Truly distributed Kalman filtering arrived with diffusion~\cite{diffusion} and consensus~\cite{average-consensus-journal,ConsFilter,average-consensus} frameworks for information fusion over sensor networks~\cite{RezaSN,CooperativeControl,SayedBook,cattivelli2010diffusion,ControlLetter,CCTA}. In these schemes each node fuses its local estimate -- and, in the information variants, its information matrix -- with its neighbours', by diffusion averaging or by consensus iterations, and approaches the performance of a centralized filter while never communicating beyond one hop. The linear theory has since matured considerably. It now includes consensus on information matrices with stability guarantees under collective observability \cite{BattistelliChisci2014,BattistelliChisci2015TAC}, covariance-intersection fusion for unknown inter-node correlations \cite{HuXieZhang2012TSP,wei2018,SebastianECODKF2024}, consistent estimation under global observability \cite{HeXueFang2018Automatica}, and distributed observers and steady-state implementations \cite{MitraSundaram2018TAC,Yan2023TAC,YangQianDuanSun2025}; see \cite{Rego2019ARC} for a survey. The frontier now includes joint state--unknown-input estimation without global topology information \cite{LuoSuZeng2026} and, in power networks, operation with only a partially known grid model \cite{kalman2019a}. All of these results assume linear dynamics and observations. Nonlinear practice is served by extended/unscented variants~\cite{UnscentedDistKalman,AliExtended,WeightAveCons,LiResilientTCST2023}, which in essence swap the local linear filter for an EKF or UKF. Boundedness results for such filters do exist. For the consensus \emph{extended} Kalman filter, stability under collective observability was established in \cite{BattistelliChisci2016EKF}, and related state-constrained distributed Kalman filters have also been studied \cite{LvDuanDuan2021}. For unscented variants, mean-square boundedness is available for the weighted-average consensus UKF \cite{WeightAveCons,wei2018} and, more recently, for a diffusion UKF with covariance intersection under intermittent measurements \cite{ChenWang2021Automatica}. However, the available analyses of these unscented variants \emph{assume a priori} the existence of uniformly bounded error covariances and of bounded compensation matrices for the linearization error (the $\Gamma,Z$ of \cref{sec:linearization}). Bounded compensations are assumed here as well (\cref{ass:bounds}); the covariance bounds, in contrast, are derived from a verifiable system-level condition. Moreover, all of the above address single-rate networks whose nodes are not severely locally unobservable. What is still lacking is an analysis that \emph{derives} covariance boundedness for statistically linearized filters from a collective detectability condition, under multi-rate sampling, and that clarifies which fusion mechanisms the guarantees do and do not cover. Providing one is a central aim of this paper.

A prominent motivation for this development is the monitoring of power networks. There, synchrophasor and other sampled-value measurements naturally arrive on different time scales, so heterogeneous -- and not necessarily integer-ratio -- sampling is operationally relevant. Kalman-type dynamic state estimation has become important for monitoring, control, and protection \cite{ZhaoTaskForce2021}, even though no individual sensor comes close to rendering the system observable. The present paper develops the estimation theory in its general form; the power-system instantiation of these methods, including transmission- and distribution-grid case studies, is reported in a companion application paper in preparation.

In this work, the problem of distributed state estimation for nonlinear networked systems under heterogeneous sampling is revisited through a merger of concepts from classical distributed filtering, nonlinear state estimation, and multi-rate sensor fusion. The contributions are fourfold. First, a multi-rate distributed filtering architecture is derived in which local continuous--discrete (comminuted, CUKF) unscented Kalman filters run on a common base time grid, perform measurement updates at each node's own sampling instants, and diffuse their results within one-hop neighbourhoods (\cref{sec:dcukf}). Second, the notion of collective uniform detectability is formalized, and three complementary strategies are developed for the setting in which individual nodes are locally unobservable: estimation on locally observable subspaces, diffusion of information matrices, and a fixed-attenuation $H_{\infty}$ information update (\cref{sec:strategies}). Third, a performance analysis is established. It provides uniform bounds on the quotient error covariances for either fused-mean rule and, when the collectively invisible subspace is trivial, exponential mean-square error bounds: without a contraction--mixing condition for the information-weighted mean, and under such a condition for the diffusion-type mean, in both cases subject to an explicit statistical-linearization coherence bound. The analysis also states explicitly which fusion mechanisms the guarantees cover, and which provably admit no such guarantee (\cref{sec:analysis}). Fourth, the framework is validated on two canonical benchmarks: the stochastic nonlinear testbed of the weighted-average consensus UKF \cite{WeightAveCons} and a three-inertia network with individually undetectable nodes \cite{KimLeeShim2020,LuoSuZeng2026}. The studies include non-integer-ratio multi-rate stress tests, per-step communication accounting, and covariance-boundedness diagnostics (\cref{sec:sim}). The studies corroborate that mean-only mixing genuinely requires a condition beyond collective detectability: with individually undetectable nodes, state-only fusion is unstable or inferior to no fusion at all, whereas information-matrix diffusion remains uniformly bounded and accurate.

The remainder of the paper is organized as follows. \cref{Sec:Problem} formulates the network and filtering problem. \cref{sec:dcukf} derives the multi-rate diffusion CUKF and its statistical linearization. \cref{sec:strategies} presents the detectability notions and the three strategies for locally unobservable nodes, whose guarantees are established in \cref{sec:analysis} and compared in \cref{sec:compare}. \cref{sec:sim} reports the numerical studies and \cref{sec:conclusion} concludes. The CUKF stages and all proofs are collected in Appendices A--B.

\noindent\textit{\textbf{Mathematical Notation}}: Scalars, column vectors, and matrices are lowercase, bold lowercase, and bold uppercase; $\mathbf{I}$ is the identity matrix of appropriate dimension. The transpose and statistical expectation operators are denoted by $(\cdot)^{\T}$ and $\E{\cdot}$. For symmetric matrices, $\mathbf{A}\preceq\mathbf{B}$ denotes that $\mathbf{B}-\mathbf{A}$ is positive semi-definite. The operators $\col_{l}\{\cdot\}$ and $\blkdiag_{l}\{\cdot\}$ stack their arguments columnwise and block-diagonally, and $\|\cdot\|$ is the spectral norm. The state dimension is $n_{\mathbf{x}}$ and the measurement dimension at node $l$ is $p_{l}$ (co-located sensors sharing the same sampling instants are stacked into one measurement vector).

\section{Problem Formulation}\label{Sec:Problem}
Sensors and communication links form the graph $\mathcal G=(\mathcal N,\mathcal E)$; $\mathcal N_l$ contains node $l$ and its communication neighbours. The nonlinear process and sensor-$l$ samples are
\begin{align}
\dot{\mathbf x}(t)&=\mathscr F(\mathbf x(t),\boldsymbol\omega(t)),\label{eq:state}\\
\mathbf y_{l,n_l}&=\mathscr H_l(\mathbf x(n_lT_l))+\boldsymbol\nu_{l,n_l},\label{eq:observe2}
\end{align}
where the process and mutually independent measurement noises are zero-mean white Gaussian. Let $\mathcal F_n$ denote the natural filtration generated by the initial state, all process and measurement noises up to base tick $n$, and the filter variables constructed from them. The objective is a fully distributed full-state estimate: each node filters locally and communicates inside $\mathcal N_l$.

\section{Multi-Rate Distributed Unscented Kalman Filtering}\label{sec:dcukf}
\subsection{Base grid and comminuted UKF}\label{sec:cukf}
For commensurate periods, $\Delta=\gcd_lT_l$ defines a base grid containing every sampling instant. At every node, the local UKF predicts at each base tick and updates only at the node's own sampling instants; \cref{Fig1} illustrates $T_l/\Delta\in\{2,3,5\}$.
\begin{figure}
    \centering
    \includegraphics[width=.68\columnwidth]{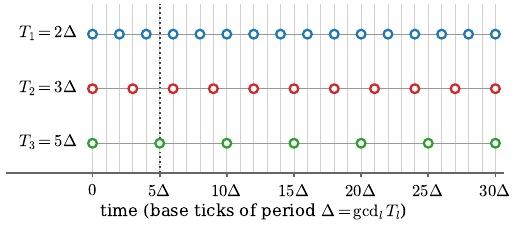}
    \caption{Multi-rate sampling on the base grid $\Delta$: circles are sensor sampling instants; the dotted guide marks a tick at which only the slowest sensor samples. The comminuted sub-steps are internal to each prediction and are not shown.}
    \label{Fig1}
\end{figure}
The comminuted UKF (CUKF) is a continuous--discrete UKF whose prediction advances one base period by $\mathsf m$ internal integration sub-steps of length $T_{\mathrm{CUKF}}=\Delta/\mathsf m$, reducing coarse-discretization error~\cite{Axelsson2015DiscreteTimeST}; $\mathsf m=1$ integrates directly on the base grid. The sub-steps are internal to each prediction, so \cref{Fig1} shows only the base grid. A measurement update is executed only if a sample is present; otherwise prediction is the complete recursion. Appendix~\ref{app:ukf} records the sigma-point stages.

\subsection{Statistical linearization and standing assumptions}\label{sec:linearization}
For a symmetric rule with $N_\sigma$ sigma points and covariance weights $W_i^{(c)}$, statistical linear regression of the CUKF points~\cite{commentson2002,UnscentedDistKalman} gives the computed matrices
\begin{equation}\label{eq:Hpseudo} {{H}_{l,n_l}}\triangleq \left ({\mathcal P_{\mathbf{x}_{n_l}\mathbf{y}_{l,n_l}}}\right )^{\T} \left ({\mathcal P_{l,n_l\mid n_l-1}}\right )^{-1} \end{equation}
and~\cite{WeightAveCons}
\begin{equation}\label{eq:Fpseudo} F_{l,n_l-1} \triangleq \left ({\mathcal P_{l,\mathbf{x}_{n_l-1}\mathbf{x}_{n_l \mid n_l-1}}}\right )^{\T} \left ({\mathcal P_{l,n_l-1|n_{l}-1}}\right )^{-1} \end{equation}
where
\begin{multline*}
{\mathcal P_{l,\mathbf{x}_{n_l-1}\mathbf{x}_{n_l \mid n_l-1}}}=\\ \sum_{i=0}^{N_\sigma-1}W_i^{(c)}(\tilde{\mathbf{x}}^{i}_{l,n_l-1|n_l-1}-\hat{\mathbf{x}}_{l,n_l-1|n_l-1})\times(\hat{\mathbf{x}}^{i}_{l,n_l|n_l-1}-\hat{\mathbf{x}}_{l,n_l|n_l-1})^{\T},
\end{multline*}
Here $\mathcal P_{l,n_l\mid n_l-1}$ is the covariance of the prediction error $\tilde{\mathbf x}_{l,n_l|n_l-1}$, and ${\mathcal P_{\mathbf{x}_{n_l}\mathbf{y}_{l,n_l}}}$ is the predicted state--measurement cross-covariance ($\mathcal P^{xy}$ in the stage summary of Appendix~\ref{app:ukf}). Two positive-semidefinite regression residuals recur below: the prediction residual $\Sigma^{f}_{l,n}=\mathcal P_f^{\sigma}-F\mathcal PF^{\T}\succeq0$, where $\mathcal P_f^{\sigma}$ is the weighted spread of the propagated sigma points about the predicted mean (equivalently $\mathcal P_{l,n_l|n_l-1}-Q_{n_l-1}$), and the measurement residual $\Sigma_{l,n}=\mathcal P_{yy}-R-H\mathcal P^{-}H^{\T}\succeq0$; they quantify what the best linear fit leaves unexplained. The accuracy of $F_{l,n}$ and $H_{l,n}$ depends on how well these regression maps approximate the nonlinear functions, and first-order Taylor arguments are not rigorous here. Following~\cite{boutayeb1997convergence,karvonen2014stability}, we instead work with \emph{exact} error relations: there exist unknown, bounded, random matrices $\Gamma_{l,n}$ and $Z_{l,n}$ -- compensation factors that absorb the linearization error and are never computed by the filter -- such that the prediction error and the innovation satisfy, exactly,
\begin{align}
\tilde{\mathbf{x}}_{l,n_l+1|n_l}&= \Gamma_{l,n_l}F_{l,n_l}\,\tilde{\mathbf{x}}_{l,n_l|n_l}+\mathbf{w}_{l,n_l}\label{eq:errdyn}\\
\boldsymbol{\epsilon}_{l,n_l}&=Z_{l,n_l} H_{l,n_l}\, \tilde{\mathbf{x}}_{l,n_l|n_l-1}+\mathbf{v}_{l,n_l}\label{eq:errobs}
\end{align}
The filter itself uses the algebraic sigma-point identities
\begin{align} {\mathcal P_{l,n_l\mid n_l-1}}=&\,F_{l,n_l-1} \mathcal P_{l,n_l-1} F_{l,n_l-1}^{\T}+Q_{n_{l}-1}+\Sigma^{f}_{l,n_l-1} \label{eq:Ppred}\\
\mathcal P_{l,n_l}=&\left ({I-\mathcal K_{l,n_l} H_{l,n_l}}\right ){\mathcal P_{l,n_l\mid n_l-1}} \label{eq:Pup}\\
\mathcal K_{l,n_l}=&\,\mathcal  P_{l,n_l\mid n_l-1} H_{l,n_l}^{\T}\left [{H_{l,n_l} \mathcal P_{l,n_l\mid {n_l-1}} H_{l,n_l}^{\T}+R_{l,n_l}+\Sigma_{l,n_l}}\right ]^{-1}\! \label{eq:Kgain}\end{align}
Thus $\Gamma$ occurs only in the true error relation through $E=(\Gamma-I)F$, never in the computed covariance, and the exact computed information recursion is
\begin{equation}\label{eq:Jinc}
\begin{split}
\mathcal P_{l,n_l}^{-1}&=\mathcal P_{l,n_l|n_l-1}^{-1}+J_{l,n_l},\\
J_{l,n_l}&\triangleq H_{l,n_l}^{\T}\big(R_{l,n_l}+\Sigma_{l,n_l}\big)^{-1}\!H_{l,n_l}\succeq0 .
\end{split}
\end{equation}
The error analyses treat $E=(\Gamma-I)F$ and $D=(Z-I)H$ as explicit perturbations; both vanish for linear maps. No structure of $\Gamma,Z$ is used anywhere in the analysis -- only the norm bounds on $E$ and $D$ in \cref{ass:bounds} enter -- so the diagonal construction of~\cite{boutayeb1997convergence,karvonen2014stability} is one admissible choice, not a requirement. The standing conditions are:
\begin{assumption}\label{ass:bounds}
For all nodes $l$ and base-grid ticks $n$:
\begin{enumerate}
    \item There exist scalars $\bar f,\check f,\bar h,\bar d_E,\bar d_D,\underline q,\bar q,\underline r,\bar r\in\mathbb{R}^{+}$ such that, almost surely,
 \begin{align}
\|F_{l,n}\|\leq\bar f,\quad \|F_{l,n}^{-1}\|\leq\check f,&\quad \|E_{l,n}\|\leq\bar d_E,\label{eq:Fbounds}\\
\|H_{l,n}\|\leq \bar h,&\quad \|D_{l,n}\|\leq\bar d_D, \label{Ccondition}\\
\underline q I\preceq Q_n\preceq Q_n+\Sigma^f_{l,n}\preceq\bar q I,&\qquad
\underline r I\preceq R_{l,n}\preceq\bar r I. \label{eq:QRbounds}
\end{align}
    Hence $\Gamma F=F+E$ and $ZH=H+D$, so the bounds on $F,H,E,D$ bound the compensated maps as well and no separate bounds on $\Gamma,Z$ are needed. Covariance results require boundedness only; the error theorems impose the stated smallness thresholds.
    \item The initial condition $\mathcal P_{l,0}$ is positive definite.
    \item The system is collectively uniformly detectable in the sense of \cref{def:collective} below.
\end{enumerate}
\end{assumption}
These uniform bounds are standard in nonlinear Kalman-filter stability analyses~\cite{boutayeb1997convergence,li2012stochastic,wei2018}.

\subsection{The multi-rate diffusion CUKF}\label{sec:alg1}
We first diffuse only the local state estimates, as in diffusion Kalman filtering~\cite{cattivelli2010diffusion,sayed2014adaptive}. Let $\mathbf{C}=[c_{j,l}]$ denote a left-stochastic combination-weight matrix compatible with the graph, i.e.,
\begin{equation}\label{Adiff}
c_{j,l}\geq 0,\qquad c_{j,l}=0 \ \text{if}\ j\notin\mathcal{N}_{l},\qquad \sum_{j\in\mathcal{N}_{l}}c_{j,l}=1 .
\end{equation}
The update at node $l$ is
\begin{equation}\label{eq:diffusion}
\hat{\mathbf{x}}_{l,n_{l}|n_{l}} \leftarrow \sum_{j\in \mathcal{N}_{l}}c_{j,l}\,\boldsymbol{\psi}_{j,n_{l}}
\end{equation}
where $\boldsymbol{\psi}_{j,n_{l}}$ is the intermediate (post-update) estimate at neighbour $j$.\begin{algorithm}[!tb]
\caption{The multi-rate distributed CUKF with state diffusion}\label{DCUKF1}
For the model \eqref{eq:state}--\eqref{eq:observe2} with weights as in \eqref{Adiff}, initialize $\hat{\mathbf{x}}_{l,0|-1}=\E{\mathbf{x}_{0}}$ and $\mathcal P_{l,0|-1}=\mathcal P_0$ at every node; offline, agree on the base period $\Delta=\gcd_{l\in\mathcal{N}}(T_{l})$ (commensurate sampling periods assumed), so that all filters tick on the base grid $t=n\Delta$. At every tick $n$, every node $l$:

Step 1: Run the CUKF prediction step (Appendix~\ref{app:ukf}) with the selected symmetric sigma-point rule over one base period via $\mathsf{m}$ internal sub-steps of length $T_{\mathrm{CUKF}}=\Delta/\mathsf{m}$, obtaining $\boldsymbol{\psi}_{l,n}\!\leftarrow\!\hat{\mathbf{x}}_{l,n|n-1}$ and $\mathcal{P}_{l,n|n-1}$.

Step 2: If a measurement is available at node $l$ at this tick, perform the CUKF measurement update (Appendix~\ref{app:ukf}), refreshing $\boldsymbol{\psi}_{l,n}$ and $\mathcal{P}_{l,n|n}$.

Step 3: Exchange $\boldsymbol{\psi}_{j,n}$ within $\mathcal{N}_{l}$ and apply the diffusion update \eqref{eq:diffusion} (covariances are not exchanged).
\end{algorithm}

The first stability result concerns the local CUKF itself.
\begin{theorem}\label{stablemma}
Let \cref{ass:bounds}.1--\ref{ass:bounds}.2 hold and suppose the local computed covariance satisfies $\underline p I\preceq\mathcal P_{l,n}^{-1}\preceq\bar p I$ for all $n$. Form the constants
\begin{align*}
\bar p^-&\triangleq\bar f^2\underline p^{-1}+\bar q,
&\bar K&\triangleq\bar p^-\bar h/\underline r,\\
\sigma_0&\triangleq\min\{1,(\bar p\bar p^-)^{-1}\},
&r_0&\triangleq(1+\delta\underline p)^{-1},\\
\delta&\triangleq\frac{\sigma_0^2\underline q}
 {\bar f^2(1+\bar K\bar h)^2},
&\bar d_*&\triangleq(1+\bar K\bar h)\bar d_E+\bar K\bar d_D(\bar f+\bar d_E).
\end{align*}
If
\begin{equation}\label{eq:localdiscrepancy}
\bar d_*\leq\varepsilon_*\triangleq
\sqrt{\frac{\underline p}{\bar p}}\,\frac{1-\sqrt{r_0}}{2},
\end{equation}
then the estimation error of the CUKF is exponentially bounded in mean square. A valid per-tick noise-free contraction factor is
$r_{\mathrm{loc}}\triangleq((1+\sqrt{r_0})/2)^2<1$. Define
\[
\tau_{\mathrm{loc}}\triangleq\frac{1-r_{\mathrm{loc}}}{2r_{\mathrm{loc}}},
\qquad
\rho_{\mathrm{loc}}\triangleq(1+\tau_{\mathrm{loc}})r_{\mathrm{loc}}=\frac{1+r_{\mathrm{loc}}}{2}<1.
\]
Then $\rho_{\mathrm{loc}}$ is a valid stochastic Lyapunov drift factor without requiring the nonlinear discrepancy matrices to be conditionally independent of the contemporaneous noises. The discrepancy condition is vacuous for linear dynamics and measurements, for which $E=D=0$.
\end{theorem}
The proof is in Appendix~\ref{app:proofs}. Existing multi-rate fusion schemes generally require the system to be observable to every sensor class, or confine fusion to the intersection of the per-class observable subspaces~\cite{bar1986effect,julier2009general,Ghosalfusion2019,abooshahab2020}; the following strategies relax that requirement.

\section{Distributed Estimation Under Collective Detectability}\label{sec:strategies}
State-only diffusion leaves each covariance local and therefore cannot control directions invisible to that node. We consider three remedies: filtering only locally observable coordinates~\cite{SinghPalDecentralized,ZhaoMili2018TSG}, diffusing information matrices, and using a fixed-attenuation $H_{\infty}$ information update. Their guarantees and costs are summarized in \cref{sec:analysis,sec:compare}.

\subsection{Detectability notions}\label{sec:detectability}
In this paper, we consider the following detectability notions \cite{conte2007algebraic}.
\begin{definition}\label{def:local}
 A sub-system \eqref{eq:state}-\eqref{eq:observe2} is said to be detectable if there exists an open and dense subset $M \subset \mathbb{R}^{n_{\mathbf{x}}}$ such that system \eqref{eq:state} is locally weakly detectable \cite{hermann1977nonlinear} at any initial state $\mathbf{x}_0 \in M$.
\end{definition}
Informally, the condition below asks that, over every window of $N$ base ticks, each node's one-hop neighbourhood accumulate at least $\beta$ of information in every direction outside one fixed, network-wide blind subspace $\mathcal V$.
\begin{definition}[Uniform one-hop neighbourhood collective detectability]\label{def:collective}
For node $l$, let
$J^{\mathrm{col}}_{l,n}\triangleq\sum_{j\in\mathcal N_l(n)}J_{j,n}$
be the sum of the computed information increments \eqref{eq:Jinc} available in its neighbourhood at base tick $n$ (an empty sum is zero). Let $\Phi_l(n,k)\triangleq F_{l,n-1}\cdots F_{l,k}$, with $\Phi_l(n,n)=I$, denote the transition product of node $l$'s computed statistical linearizations (a product of $n-k$ factors), and write $\Phi_l(k,n)\triangleq\Phi_l(n,k)^{-1}$. The system is called \emph{uniformly one-hop collectively detectable} if there exist an integer $N\geq1$, a constant $\beta>0$, and a fixed subspace $\mathcal V\subseteq\mathbb R^{n_{\mathbf x}}$ with $\Phi_l(k,n)\mathcal V=\mathcal V$ and $J^{\mathrm{col}}_{l,k}\Phi_l(k,n)v=0$ for every $v\in\mathcal V$, such that, with $\Pi$ the orthogonal projector onto $\mathcal V^\perp$, for every node $l$ and every $n\geq N$,
\begin{equation}\label{eq:gramian}
\sum_{k=n-N+1}^{n}\Phi_l(k,n)^{\T}J^{\mathrm{col}}_{l,k}\Phi_l(k,n)\succeq\beta\Pi.
\end{equation}
For $\mathcal V\neq\{0\}$ the pair is thus collectively detectable on $\mathcal V^{\perp}$ only, which is all that is required in the sequel.
\end{definition}
\begin{remark}\label{rem:def}
Four reading notes on \cref{def:collective}. \emph{(i) Relation to the classical condition.} For $\mathcal V=\{0\}$, \eqref{eq:gramian} reduces to the usual uniform information-Gramian requirement, imposed along the trajectory the filter actually computed~\cite{boutayeb1997convergence,li2012stochastic}. \emph{(ii) What the invisible subspace costs.} For nontrivial $\mathcal V$, nothing is claimed along $\mathcal V$ -- modes there may be neutral or unstable -- and every guarantee is a statement on the factor space $\mathbb R^{n_{\mathbf x}}/\mathcal V$, the standard quotient construction of geometric control theory~\cite{wonham1985,conte2007algebraic}. \emph{(iii) Realized, not nominal.} The condition constrains the realized sequences $F,J$, not a design model: a lost packet is simply a tick with $J=0$, and the condition is evaluated numerically in \cref{sec:sim}. A design-time observability certificate for the nominal model transfers to the realized condition only together with a quantitative perturbation margin. \emph{(iv) Place in the literature.} Because \eqref{eq:gramian} must hold at every node from its own one-hop information, the condition is stronger than a single graph-wide requirement; local, distributed, or weighted uniform detectability~\cite{anderson1981detectability,cattivelli2010diffusion,li2018weightedly} each imply it, but not conversely.\phantomsection\label{rem:collectivedet}
\end{remark}
The definition extends the collective/global observability conditions of linear distributed filtering~\cite{BattistelliChisci2014,BattistelliChisci2015TAC,wei2018,HeXueFang2018Automatica} to realized, multi-rate statistical linearizations; in this windowed, realized form it is, to our knowledge, new.
\begin{lemma}\label{lem:rank}
Let $\mathcal O_{l,n}$ be a finite-window observability matrix of the realized \emph{computed} pair $(F_{l,n},J_{l,n}^{1/2})$, where any positive-semidefinite square root may be used. (i)~If $\mathcal O_{l,n}^{\T}\mathcal O_{l,n}\succeq\beta_l I$ with $\beta_l>0$ independent of $n$, the pair is uniformly completely observable~\cite{anderson1981detectability}. (ii)~If the nullspace of $\mathcal O_{l,n}$ is fixed and invariant under $F_{l,n}$, then any fixed, uniformly conditioned $\Theta_l$ built from it gives a block-triangular observable/unobservable decomposition, and the observable block is uniformly completely observable whenever its reduced Gramian is uniformly bounded below.
\end{lemma}
However, rank alone does not suffice for (i): observability strength may decay with $n$. For a time-varying nullspace, a time-varying $\Theta_{l,n}$ and the correspondingly transformed dynamics are required. The nonlinear Lie-derivative rank condition~\cite{conte2007algebraic,nijmeijer1990nonlinear} guides offline construction of $\Theta_l$ but is distinct from the realized LTV matrix in \cref{lem:rank}; only the latter enters the analysis.\phantomsection\label{rem:hk}

\subsection{Strategy 1: locally observable subspaces}\label{sec:alg2}
For use in the simplification of the local estimation, node $l$ maintains the observability decomposition
$$\begin{bmatrix}
\hat{\mathbf{x}}^{ob}_{l,n_l+1|n_l}\\[2pt] \hat{\mathbf{x}}^{uob}_{l,n_l+1|n_l}
\end{bmatrix}=\Theta_{l}\,\hat{\mathbf{x}}_{l,n_l+1|n_l}$$
with a nonsingular transformation $\Theta_{l}$ obtained from the decomposition of \cref{lem:rank}. Each node then runs the CUKF only on its observable coordinates, as summarized in \cref{DCUKF2}.
\begin{algorithm}[!tb]
\caption{Diffusion multi-rate CUKF on locally observable subspaces}\label{DCUKF2}
For each node $l$ and its neighbourhood $\{j\in \mathcal N_l\}$, repeat at every base-grid tick $n$:
   \begin{enumerate}
  \item (Offline) agree on the base period $\Delta=\gcd_{l}(T_{l})$ and construct each local transformation $\Theta_l$. All communicated estimates are reconstructed in the common full-state coordinates before diffusion, so no inter-node coordinate map is used online.
\item Transform the current estimate as $[\hat{\mathbf{x}}^{ob,\T}_{l,n|n},\,\hat{\mathbf{x}}^{uob,\T}_{l,n|n}]^{\T}=\Theta_{l}\hat{\mathbf{x}}_{l,n|n}$.
\item Apply the CUKF prediction step (Appendix~\ref{app:ukf}) to the reduced dynamics $\mathscr F^{ob}$ with the corresponding reduced sigma-point set, obtaining $\hat{\mathbf{x}}^{ob}_{l,n+1|n}$ and $\mathcal P^{ob}_{l,n+1|n}$; propagate the unobservable coordinates by prediction only.
\item Initialize $\boldsymbol\psi^{ob}_{l,n+1}\leftarrow\hat{\mathbf{x}}^{ob}_{l,n+1|n}$ and $\mathcal P^{ob}_{l,n+1}\leftarrow\mathcal P^{ob}_{l,n+1|n}$. If node $l$ has a measurement at tick $n+1$, perform the reduced CUKF measurement update. Then reconstruct
$$
\boldsymbol{\psi}_{l,n+1}=\Theta_l^{-1}\begin{bmatrix}
\boldsymbol{\psi}^{ob}_{l,n+1}\\[2pt] \boldsymbol{\psi}^{uob}_{l,n+1}
\end{bmatrix}.
$$
\item Node $j$ communicates the full-coordinate estimate $\boldsymbol\psi_{j,n+1}$; the reduced covariance $\mathcal P^{ob}_{j,n+1}$ remains local. For $\sum_{j\in\mathcal N_l}c_{j,l}=1$, set
$$
\hat{\mathbf{x}}_{l,n+1|n+1}\leftarrow\sum_{j\in\mathcal N_l}c_{j,l}\boldsymbol\psi_{j,n+1}.
$$
   \end{enumerate}
\end{algorithm}

\subsection{Strategy 2: diffusing information matrices}\label{sec:alg3}
An alternative is to diffuse the covariances as well. This is similar in spirit to the weighted-average consensus UKF of \cite{WeightAveCons}, with two differences: a single fusion step per base tick (measurement update and fusion in one step), and averaging on the \emph{information} matrices $\mathcal P^{-1}$ rather than the covariances. It is this information-form choice that reinstates the neighbourhood information term of \eqref{eq:inforecur} and thereby enables \cref{thm:Pbounds3}. The diffusion-based approach for both state estimates and covariances is summarized in \cref{DCUKF3}. Choosing proper diffusion weights $c_{j,l}$ can improve the estimator; the Bar-Shalom--Campo \cite{Ghosalfusion2019} or covariance-intersection \cite{HuXieZhang2012TSP,abooshahab2020} methods obtain such weights more systematically, a point we return to in \cref{sec:analysis,sec:sim}. Notably, a single fusion exchange per interval is exactly the regime in which CI-type consensus retains stability, whereas consensus on measurements requires sufficiently many exchanges per interval to secure detectability from the $L$-hop neighbourhood \cite{BattistelliChisci2015TAC}.
\begin{algorithm}[!tb]
\caption{One-step multi-rate DCUKF with state and information diffusion}\label{DCUKF3}
Consider the state-space model \eqref{eq:state}-\eqref{eq:observe2} and combination weights as in \eqref{Adiff}.

\noindent
Initialization at every node $l$:
$$ \hat{\mathbf{x}}_{l,0|-1}=\E{\mathbf{x}_{0}},\qquad \mathcal P_{l,0|-1}=\mathcal P_0 .$$
Offline: agree on the base period $\Delta=\gcd_{l\in\mathcal{N}}(T_{l})$. At every tick $n$ of the base grid, at every node $l$:

Step 1: Run the CUKF prediction step (Appendix~\ref{app:ukf}) over one base period with $\mathsf{m}$ internal CUKF sub-steps.

Step 2: If a measurement is available at node $l$ at this tick, perform the CUKF measurement update (Appendix~\ref{app:ukf}); denote the results $\boldsymbol{\psi}_{l,n}$, $\mathcal{P}_{l,n}$.

Step 3: For $\sum_{j\in \mathcal N_l} c_{j,l}=1$, exchange $(\boldsymbol{\psi}_{j,n},\mathcal P_{j,n}^{-1})$ within $\mathcal N_{l}$ and fuse the covariances,
$$
(\mathcal P_{l,n|n})^{-1} \leftarrow \textstyle\sum_{j\in \mathcal N_l}c_{j,l}\,(\mathcal P_{j,n})^{-1},
$$
with the fused mean given by either the \emph{diffusion mean}
\begin{equation}\label{eq:fusemeana}
\hat{\mathbf{x}}_{l,n|n} \leftarrow \textstyle\sum_{j\in \mathcal N_l}c_{j,l}\,\boldsymbol{\psi}_{j,n}
\end{equation}
or the \emph{information mean}
\begin{equation}\label{eq:fusemeanb}
\hat{\mathbf{x}}_{l,n|n} \leftarrow \mathcal P_{l,n|n}\textstyle\sum_{j\in \mathcal N_l}c_{j,l}\,(\mathcal P_{j,n})^{-1}\boldsymbol{\psi}_{j,n}.
\end{equation}
\end{algorithm}

\subsection{Strategy 3: a fixed-attenuation \texorpdfstring{$H_{\infty}$}{H-infinity} information update}\label{sec:alg4}
The third strategy adapts the unscented $H_{\infty}$ information filter~\cite{LiJia2010UHinf} to one-hop, multi-rate operation. Unlike prior all-to-all or decentralized implementations with per-step attenuation tuning~\cite{LiJia2010UHinf,ZhaoMili2018TSG,ZhaoMili2019TSP}, \cref{UKF22,ass:hinf,thm:hinfP} use a fixed $\theta$ and give a uniform well-posedness condition under collective detectability.

In the information form, the neighbourhood information contributions at node $l$ are
\begin{align}
\mathcal I_{l,n_l}&= \textstyle\sum_{j\in\mathcal{N}_{l}(n_l)}\mathcal P^{-1}_{j,n_l|n_l-1}\mathcal P^{xy}_{j,n_l|n_l-1}R^{-1}_{j,n_l}\times({\mathcal P^{xy}_{j,n_l|n_l-1}})^{\T}\mathcal P^{-1}_{j,n_l|n_l-1},\label{eq:hinfI}\\
\mathcal T_{l,n_l}&= \textstyle\sum_{j\in\mathcal{N}_{l}(n_l)}\mathcal P^{-1}_{j,n_l|n_l-1}\mathcal P^{xy}_{j,n_l|n_l-1}R^{-1}_{j,n_l}\times(\mathbf{y}_{j,n_l}-\mathscr{H}_{j}(\hat{\mathbf{x}}_{l,n_l|n_l-1})),\label{eq:hinfT}
\end{align}
where $\mathcal P^{xy}_{j}$ denotes the state--measurement cross-covariance of the CUKF update (Appendix~\ref{app:ukf}) at node $j$. The innovation in \eqref{eq:hinfT} is evaluated at node $l$'s predicted estimate, while the weights are calibrated at node $j$'s. The exact-relation apparatus extends to such cross evaluations: for each pair there is a bounded $Z_{j|l,n}$ -- the cross-evaluation analogue of the compensation matrices $Z$ of \cref{sec:linearization} -- with $\mathbf{y}_{j,n}-\mathscr H_{j}(\hat{\mathbf{x}}_{l,n|n-1})=Z_{j|l,n}H_{j,n}\tilde{\mathbf{x}}_{l,n|n-1}+\mathbf{v}_{j,n}$ and the resulting weight/relation mismatch is collected in the \emph{cross-evaluation mismatch}
\begin{equation}\label{eq:crossdelta}
\Upsilon_{l,n}\triangleq\textstyle\sum_{j\in\mathcal{N}_{l}(n)}H_{j,n}^{\T}R_{j,n}^{-1}\big(Z_{j|l,n}-I\big)H_{j,n},
\end{equation}
which vanishes for linear measurement maps ($Z_{j|l}=Z_{j}=I$). Also, $\mathcal P^{-1}\mathcal P^{xy}R^{-1}({\mathcal P^{xy}})^{\T}\mathcal P^{-1}=H^{\T}R^{-1}H$ by \eqref{eq:Hpseudo}. The $H_\infty$ measurement update subtracts the fixed attenuation penalty $\theta^{-2}I$ from the stored information
\begin{equation}\label{eq:hinfup}
\mathcal P_{l,n_l\vert n_l}^{-1}=\mathcal P_{l,n_l\vert n_l-1}^{-1}+\mathcal I_{l,n_l}-\theta^{-2}I ,
\end{equation}
the intermediate estimate is reconstructed as
\begin{equation}\label{eq:hinfpsi}
\boldsymbol{\psi}_{l,n_l}=\hat{\mathbf{x}}_{l,n_l\vert n_l-1}+\big(\mathcal P_{l,n_l\vert n_l}^{-1}+\theta^{-2}I\big)^{-1}\mathcal T_{l,n_l} ,
\end{equation}
and the estimates are then combined by the diffusion update \eqref{eq:diffusion}.
\begin{algorithm}[!tb]
\caption{Distributed nonlinear $H_\infty$ information filter}\label{UKF22}
For the system model \eqref{eq:state}-\eqref{eq:observe2}, starting with $\hat{\mathbf{x}}_{l,0\vert-1}=\E{\mathbf{x}_{0}}$ and $\mathcal P_{l,0\vert-1}=\mathcal P_{0}$; offline, agree on the base period $\Delta=\gcd_{l}(T_{l})$. At every tick of the base grid, at every node $l$:

Step 1: Run the CUKF prediction step (Appendix~\ref{app:ukf}) on the base grid.

Step 2: Exchange the information contributions available in the neighbourhood at this tick (empty sums if none) and perform the $H_\infty$ information update \crefrange{eq:hinfI}{eq:hinfpsi} at every base tick; in particular, the subtraction of $\theta^{-2}I$ in \eqref{eq:hinfup} is applied at every base tick.

Step 3: Perform the diffusion update \eqref{eq:diffusion}.
\end{algorithm}

\begin{assumption}[$H_{\infty}$ feasibility]\label{ass:hinf}
Let $\bar d$ be the maximum neighbourhood cardinality and, for the chosen fixed $\theta$, define
\begin{align*}
\bar p_H&\triangleq\max\!\left\{
\max_l\lambda_{\max}(\mathcal P_{l,0}^{-1}+\theta^{-2}I),
\ \underline q^{-1}+\bar d\bar h^2\underline r^{-1}\right\},\\
\alpha_H&\triangleq(1+\bar p_H\bar q\check f^2)^{-1}.
\end{align*}
and let $\underline p_0\triangleq(\alpha_H/2)^{N-1}\beta$. The attenuation level satisfies $\theta^{-2}<\underline p_0/2$, $\mathcal V=\{0\}$ in \cref{def:collective}, and the finite-prefix regularized-information margin
$M_{l,n}=\mathcal P_{l,n}^{-1}+\theta^{-2}I\succeq\underline p_0I$
holds for $0\leq n<N$ and every node. These are finite, checkable feasibility conditions for the selected attenuation level; the windowed Gramian restores the regularized-information lower margin only after a complete window and cannot replace the check over this initial prefix of ticks. The attenuation subtraction itself does not create this margin.
When $g\triangleq\alpha_H/\bar f^{2}\in(0,1)$, a conservative sufficient initialization for that finite-prefix check is
\[
M_{l,0}\succeq g^{-(N-1)}\!\left(\underline p_{0}+\frac{\theta^{-2}g}{1-g}\right)I.
\]
Because $\alpha_H$ depends on $\bar p_H$, and hence on the proposed initial information, this sufficient condition is evaluated self-consistently for the chosen initialization; direct verification of the finite prefix remains the definitive test.
\end{assumption}

\begin{remark}\label{rem:hinfquot}
The restriction $\mathcal V=\{0\}$ is intrinsic, not an artefact of the proof. Along a neutral collectively invisible direction the predicted information decays geometrically while \eqref{eq:hinfup} subtracts the fixed $\theta^{-2}I$; the regularized information therefore eventually drops below any fixed floor, $\mathcal P^{-1}$ loses positive definiteness, and the $H_{\infty}$ construction cannot be sustained on the full state space. For maps that descend to the quotient, \cref{cor:quoterr} formalizes the reduction; the $H_{\infty}$ construction likewise transfers to the quotient-reduced filter, with the projected $E$ and $D$ discrepancies, and on the quotient $\mathcal V=\{0\}$ by construction.
The discrepancy \eqref{eq:crossdelta} is a genuine feature of the diffusion $H_{\infty}$ architecture, not a proof artefact: neighbours calibrate their information contributions at their own estimates while node $l$ consumes them at its own. The cross-evaluation mismatch $\Upsilon_{l,n}$ of \eqref{eq:crossdelta} vanishes for linear channels, and for smooth channels it is of the order of the neighbourhood estimate spread. Condition \eqref{eq:hinfdiscrepancy} makes the combined dynamic and cross-evaluation coherence requirement explicit; in nonlinear applications its bound is generally certified offline from regional smoothness and estimate-spread bounds rather than observed directly by the deployed filter.\phantomsection\label{rem:crossnode}
\end{remark}

\section{Performance Analysis}\label{sec:analysis}
\begin{table*}[t]
\centering
\caption{Comparison of the proposed multi-rate distributed CUKF strategies ($n_{\mathbf{x}}$: state dimension; $n^{ob}_{\mathbf{x},l}\leq n_{\mathbf{x}}$: dimension of node $l$'s observable subspace).}
\label{tab:compare}
\footnotesize
\setlength{\tabcolsep}{4pt}
\begin{tabular}{@{}llp{1.9cm}cp{4.0cm}@{}}
\hline
Strategy & Exchanged / neighbour / tick & Cov.\ bnd & Err.\ bnd & Practical remarks\\
\hline
\cref{DCUKF1}, fixed weights & $\hat{\mathbf{x}}$ ($n_{\mathbf{x}}$) & none in general & --- & can diverge if a node is not locally detectable\\
\cref{DCUKF1}, CI weights & $\hat{\mathbf{x}}$, $\trace\mathcal P$ ($n_{\mathbf{x}}{+}1$) & diverges (companion paper) & --- & reweighting alone cannot restore unobservable directions\\
\cref{DCUKF2} (subspaces) & $\hat{\mathbf{x}}$ ($n_{\mathbf{x}}$) & \cref{prop:alg2} (reduced) & \cref{cor:alg2err} (subspace) & fixed local decompositions; smallest covariance matrices\\
\cref{DCUKF3} (state $+$ info) & $\hat{\mathbf{x}}$, $\mathcal P^{-1}$ ($\tfrac{n_{\mathbf{x}}(n_{\mathbf{x}}+3)}{2}$) & \cref{thm:Pbounds3} & \cref{thm:errbound} & recommended default (\cref{sec:sim})\\
\cref{UKF22} ($H_{\infty}$) & as \cref{DCUKF3} & \cref{thm:hinfP} & \cref{thm:hinferr} & robust to model uncertainty; needs \cref{ass:hinf}\\
\hline
\end{tabular}
\end{table*}

We now analyze the three strategies in turn -- what each one guarantees, under which conditions, and where the guarantees stop. All proofs are collected in Appendix~\ref{app:proofs}.

\subsection{The subspace filter and neighbourhood-information recursions}
\begin{theorem}\label{thm:Pbounds}
Consider the system \eqref{eq:state}-\eqref{eq:observe2}. Under \cref{ass:bounds}, any local recursion that assimilates the neighbourhood information as in \eqref{eq:inforecur} -- the family of consensus-on-information updates, used below as the comparison object for \cref{DCUKF3} -- has uniformly bounded quotient error covariances: there exist $0<\underline p\leq\bar p$ such that
$$
\underline{p}\,\Pi\preceq \mathcal P_{l,n_l}^{-1}\preceq\bar{p}\,I, \qquad n_l \geq N .
$$
\end{theorem}
\begin{proposition}[\cref{DCUKF2}]\label{prop:alg2}
Under \cref{ass:bounds}, suppose the local decompositions are uniformly regular for the computed covariance recursion: the observable block of $\Theta_lF_{l,n}\Theta_l^{-1}$ and the corresponding reduced increment $J^{ob}_{l,n}=\Theta_l^{-\T}J_{l,n}\Theta_l^{-1}$ satisfy the reduced form of \eqref{eq:gramian} with $\mathcal V=\{0\}$. Then the reduced covariances of \cref{DCUKF2} obey $\underline p_{2}I\preceq(\mathcal P^{ob}_{l,n_l})^{-1}\preceq\bar p_{2}I$. No claim is made for the unobservable block, which is propagated by prediction only.
\end{proposition}
\begin{corollary}[subspace error band for \cref{DCUKF2}]\label{cor:alg2err}
Let the hypotheses of \cref{prop:alg2} hold, and let the linearization discrepancies, measured in the reduced coordinates, stay below the thresholds of \cref{thm:errbound} formed with $(\underline p_2,\bar p_2)$, which is automatically fulfilled for linear systems. Then:
(i) \emph{before} the diffusion step, each node's error on its own observable subspace, $\tilde{\mathbf z}^{ob}_{l,n}=(\Theta_l\tilde{\mathbf x}_{l,n})^{ob}$, is exponentially bounded in mean square, with the single-node constants of \cref{stablemma} formed from the reduced quantities (rate $\rho^{(2)}_{\mathrm{loc}}$);
(ii) the bound survives the diffusion of Step~5 provided every neighbour that node $l$ listens to observes at least the same directions, $\operatorname{span}\Theta_l^{ob}\subseteq\operatorname{span}\Theta_j^{ob}$ whenever $c_{j,l}>0$, together with $\kappa_2\rho^{(2)}_{\mathrm{loc}}<1$, $\kappa_2\triangleq\bar p_2/\underline p_2$;
(iii) without that condition the bound can fail; see the check below.
The proof is in Appendix~\ref{app:proofs}.
\end{corollary}
A direct numerical check on a 14-bus test network confirms both regimes. With the diffusion step disabled, the subspace error stays within $0.035$--$0.092\,$rad across all ranks and seeds, with $\lambda_{\max}(\mathcal P^{ob})\approx1$. With the full-coordinate diffusion active, one run exhibits the injection: a neighbour's unobservable-block drift enters the observable coordinates and $\tilde{\mathbf z}^{ob}$ exceeds $10^{3}\,$rad, while $\lambda_{\max}(\mathcal P^{ob})=1.08$ stays bounded -- the covariance never sees the diffusion. The span-compatibility in (ii) is admittedly restrictive; restricting the diffusion to commonly observable coordinates would relax it, at the price of per-pair coordinate maps, and is left for future work.
\begin{remark}\label{rem:alg1caveat}
\cref{thm:Pbounds} does \emph{not} cover \cref{DCUKF1}. In \cref{DCUKF1}, each node updates with its \emph{own} measurement only and diffuses state estimates with covariance-agnostic weights. Between the (possibly rare) samples of a slow sensor, the local information recursion contains no neighbourhood term, so nothing prevents $\mathcal{P}_{l}$ from growing in the directions that are unobservable from node $l$'s own measurements. The numerical study in \cref{sec:sim} exhibits this behaviour. It moreover shows that information-aware combination weights (e.g., covariance-intersection weights $c_{j,l}\propto 1/\trace \mathcal{P}_{j}$) are \emph{not} by themselves a remedy when the locally unobservable subspaces are large: reweighting a state-only fusion cannot restore directions unobservable to the whole neighbourhood. What does restore the guarantee of \cref{thm:Pbounds} is diffusing the information matrices as in \cref{DCUKF3}, which reinstates the neighbourhood information term.
 \cref{DCUKF1} is thus the estimate-only-sharing variant; it is \emph{not} the diffusion Kalman filter of \cite{diffusion}, whose incremental step shares raw measurements across the neighbourhood and which converges under per-neighbourhood detectability.
\end{remark}
\subsection{Information diffusion}
\begin{theorem}\label{thm:Pbounds3}
Consider the system \eqref{eq:state}-\eqref{eq:observe2}. Under \cref{ass:bounds}, and provided the weights satisfy $c_{j,l}\geq\underline{c}>0$ for all $j\in\mathcal{N}_{l}$, the error covariances obtained by \cref{DCUKF3} -- with either mean rule in Step 3, which share the covariance recursion -- satisfy
$$
\underline{p}_{3}\,\Pi\preceq \mathcal P_{l,n}^{-1}\preceq\bar{p}_{3}\,I, \qquad n \geq N ,
$$
with $\bar p_{3}=\underline q^{-1}+\bar h^{2}\underline r^{-1}$, $\alpha_3=(1+\bar p_3\bar q\check f^2)^{-1}$, and $\underline p_{3}=\underline c\,(\alpha_3\underline c)^{N-1}\beta$; the upper bound holds for all $n\geq1$, and for $n<N$ the initialization supplies a positive, initialization-dependent lower constant. Consequently, for any orthonormal basis $U$ of $\mathcal V^{\perp}$, the quotient covariance obeys $U^{\T}\mathcal P_{l,n}U\preceq\underline p_{3}^{-1}I$, equivalently $\Pi\mathcal P_{l,n}\Pi\preceq\underline p_{3}^{-1}\Pi$. No uniform full-state lower bound on $\mathcal P^{-1}$ can hold when $\mathcal V\neq\{0\}$: along a collectively invisible direction the information decays, as \cref{sec:sim} exhibits. The proof is in Appendix~\ref{app:proofs}.
\end{theorem}
\begin{theorem}\label{thm:errbound}
Consider \cref{DCUKF3}. Let \cref{ass:bounds} hold, let $c_{j,l}\geq\underline c>0$, and let $\mathcal V=\{0\}$ in \cref{def:collective}. The fused covariances are bounded by \cref{thm:Pbounds3}; the local posteriors satisfy $\underline p_L I\preceq(\mathcal P^L_{l,n})^{-1}\preceq\bar p_3I$ with $\alpha_3=(1+\bar p_3\bar q\check f^2)^{-1}$ and $\underline p_L=\alpha_3\underline p_3/\bar f^2$. Use the stage-uniform pair -- one pair valid at every stage of the tick, post-fusion, post-prediction, and post-update --
\[
(\underline p,\bar p)=\bigl(\min\{\underline p_3,\underline p_L\},\bar p_3\bigr)
\]
and form $\bar K,r_0,\bar d_*,\varepsilon_*,r_{\mathrm{loc}}$, and $\rho_{\mathrm{loc}}$ exactly as in \cref{stablemma} with this pair. If $\bar d_*\leq\varepsilon_*$, then the information mean \eqref{eq:fusemeanb} has exponentially bounded mean-square error. The diffusion mean \eqref{eq:fusemeana} has the same property provided additionally
\begin{equation}\label{eq:gaincondition}
\frac{\bar p}{\underline p}\,\rho_{\mathrm{loc}}<1.
\end{equation}
\end{theorem}
\begin{remark}\label{rem:kappa}
Condition \eqref{eq:gaincondition} concerns the \emph{diffusion mean} only: it mixes vectors measured in different local metrics and therefore pays the norm-equivalence factor $\kappa=\bar p/\underline p$ on the stochastic drift factor $\rho_{\mathrm{loc}}$. The information mean pays no such mixing factor because \cref{lem:mix} is exact. Both branches still require the statistical-linearization coherence condition \eqref{eq:localdiscrepancy}; thus ``unconditional'' below always means \emph{without a contraction--mixing condition}, not without control of the nonlinear linearization residuals. Whether the diffusion mean is bounded under collective uniform detectability alone is open; in \cref{sec:sim} it diverges in the strongly nonlinear multi-rate benchmark. The covariance recursion, and hence \cref{thm:Pbounds3}, is common to both mean rules; the fusion enters the error recursion only through the mean.
 This is consistent with the unconditional convergence of the diffusion Kalman filter in the linear, single-rate regime of \cite{diffusion}, where every neighbourhood pair $\{F,H^{\mathrm{loc}}_{l}\}$ is assumed detectable and the incremental step shares raw data, so \eqref{eq:gaincondition} is not needed. The condition is the price of the strictly weaker collective-detectability, multi-rate regime targeted here, in which per-neighbourhood detectability fails and the diffusion mean can genuinely diverge (\cref{sec:sim}).
\end{remark}
\begin{corollary}[quotient error]\label{cor:quoterr}
Suppose the maps descend to the quotient: $H_nW=0$ and $U^{\T}F_nW=0$ for all $n$, where the columns of $U$ and $W$ span $\mathcal V^{\perp}$ and $\mathcal V$; for linear systems with structural $\mathcal V$ this holds automatically. Then:
(i) each node's quotient blocks $U^{\T}\hat{\mathbf x}$, $U^{\T}\mathcal PU$, $U^{\T}\mathcal K$ coincide exactly -- whatever the cross-covariance -- with \cref{DCUKF3} run on the quotient model $(U^{\T}F_nU,\,H_nU,\,U^{\T}(Q_n{+}\Sigma^{f}_n)U,\,R{+}\Sigma)$. That model satisfies \cref{def:collective} with $\mathcal V=\{0\}$ and the same $\beta$, so \cref{thm:Pbounds3} gives $\bar p_3^{-1}I\preceq U^{\T}\mathcal PU\preceq\underline p_3^{-1}I$;
(ii) for the diffusion mean \eqref{eq:fusemeana} the scalar weights commute with $U^{\T}$, so the quotient error $U^{\T}\tilde{\mathbf x}$ satisfies \cref{thm:errbound} verbatim with the quotient constants;
(iii) the same holds for the information mean \eqref{eq:fusemeanb} when $\mathcal V$ is also dynamically decoupled, $W^{\T}F_nU=0$ (e.g., structural neutral modes of ungrounded networks, \cref{sec:sim}): the cross-covariances then vanish and the mixing weights are block-diagonal.
The proof is in Appendix~\ref{app:proofs}.
\end{corollary}
\begin{remark}\label{rem:quoterr}
\cref{cor:quoterr} claims only the quotient -- the full-state error along $\mathcal V$ is an unobserved random walk. Open is the information mean with \emph{coupled} $\mathcal V$ ($W^{\T}F_nU\neq0$): the weights $M_j=\mathcal P_{\mathrm{new}}\mathcal P_j^{-1}$ acquire cross blocks, and although $\sum_jU^{\T}M_jW=0$ cancels the common $\mathcal V$-drift, the \emph{spread} of the nodes' $\mathcal V$-errors can enter the quotient. Running the filter directly on $U^{\T}\mathbf x$ sidesteps this for any descending model (\cref{thm:errbound}).
\end{remark}

\subsection{The \texorpdfstring{$H_{\infty}$}{H-infinity} filter}
The results in this subsection concern the arithmetic diffusion step implemented in \cref{UKF22}; an information-weighted state fusion would require a matching convex fusion and storage of the posterior information matrices, which is not part of that algorithm.
\begin{theorem}\label{thm:hinfP}
Under \cref{ass:bounds,ass:hinf}, the regularized information matrix of \cref{UKF22} is uniformly bounded:
\[
\underline p_0I\preceq M_{l,n}\triangleq\mathcal P_{l,n}^{-1}+\theta^{-2}I\preceq\bar p_HI,
\qquad n\geq0.
\]
Consequently $(\underline p_0/2)I\preceq\mathcal P_{l,n}^{-1}\preceq\bar p_HI$, and the update \eqref{eq:hinfup} is well posed at every base tick.
\end{theorem}
\begin{theorem}\label{thm:hinferr}
Under \cref{ass:bounds,ass:hinf}, define
\begin{align*}
\underline\omega&\triangleq\alpha_H\underline p_0/(2\bar f^2), &
\gamma_\theta&\triangleq\underline\omega/\bar p_H^2,\\
\delta_\theta&\triangleq\underline q\,\underline p_0/(2\bar f^2), &
r_{\theta,0}&\triangleq(1-\theta^{-2}\gamma_\theta)(1+\delta_\theta)^{-1},\\
\bar d_\Upsilon&\triangleq\sup_{l,n}\|\Upsilon_{l,n}\|, &
\bar d_\theta&\triangleq\underline p_0^{-1}\!\left[\underline q^{-1}\bar d_E+\bar d_\Upsilon(\bar f+\bar d_E)\right].
\end{align*}
with $\Upsilon_{l,n}$ the cross-evaluation mismatch \eqref{eq:crossdelta}. If
\begin{equation}\label{eq:hinfdiscrepancy}
\bar d_\theta\leq
\sqrt{\frac{\underline p_0}{2\bar p_H}}\,\frac{1-\sqrt{r_{\theta,0}}}{2},
\end{equation}
then the local $H_\infty$ step has the noise-free contraction factor
$r_\theta\triangleq((1+\sqrt{r_{\theta,0}})/2)^2<1$. Define
\[
\tau_\theta\triangleq\frac{1-r_\theta}{2r_\theta},
\qquad
\rho_\theta\triangleq(1+\tau_\theta)r_\theta=\frac{1+r_\theta}{2}<1.
\]
For the diffusion mean implemented in Step 3 of \cref{UKF22}, the estimation error is exponentially bounded in mean square if additionally
\begin{equation}\label{eq:hinfgain}
\kappa_\theta \rho_\theta<1,
\qquad \kappa_\theta\triangleq2\bar p_H/\underline p_0.
\end{equation}
For linear dynamics and measurement channels, \eqref{eq:hinfdiscrepancy} is automatic because $E=\Delta=0$; the fusion condition \eqref{eq:hinfgain} remains. No information-mean branch is claimed for \cref{UKF22}, because its Step 3 diffuses states without storing the convexly fused information matrix required by \cref{lem:mix}.
\end{theorem}

\section{Comparison of the Strategies}\label{sec:compare}
\Cref{tab:compare} gives the main communication and guarantee trade-offs. State-only \cref{DCUKF1} is cheapest but has no covariance guarantee. The subspace filter \cref{DCUKF2} bounds smaller local covariance matrices but needs fixed decompositions and full-state reconstruction before diffusion. Information diffusion \cref{DCUKF3} bounds quotient covariances for either mean; under the coherence condition, the information mean avoids the contraction--mixing condition required by the diffusion mean and is the recommended default. The $H_{\infty}$ filter has the same communication order and adds robustness through a feasible fixed attenuation penalty together with a proved lower margin on the regularized information.

The computed $F,H,\Sigma^f,\Sigma,J$ and realized Gramian can be logged during commissioning. The true-relation discrepancies $E,D,\Upsilon$ are not directly observable and require regional offline bounds or independent certification. The finite-prefix $H_{\infty}$ regularized-information margin is directly checkable, while the diffusion-mean mixing test is conservative.

\section{Numerical Studies}\label{sec:sim}
Two benchmark studies validate the framework outside any particular application domain. Both studies use the four-node graph and Metropolis weight matrix of \cite{WeightAveCons} (edges $1\!-\!2$, $1\!-\!3$, $2\!-\!3$, $3\!-\!4$) and Monte Carlo averaging. Within each run, every estimator receives the same truth, initial estimates, and measurement noises. All estimators share the scaled symmetric sigma set with $2n_{\mathbf{x}}{+}1$ points and $n_{\mathbf{x}}+\kappa=3$, which fits the generic weighted notation of Appendix~\ref{app:ukf}. The theoretical statements apply whenever the resulting covariances and regression residuals satisfy \cref{ass:bounds}. The centralized reference assimilates the stacked measurements sequentially with sigma points regenerated per sensor. A run is declared divergent if its error exceeds $10^{3}$ after step 20 or becomes non-finite; divergent runs are excluded from the reported RMSE and counted separately.

\subsection{Stochastic nonlinear benchmark}\label{sec:simbench1}
We reproduce the nonlinear WAC-UKF example of~\cite{WeightAveCons} with its dynamics, four noise levels, initial estimates, $100$ steps, and $30$ runs. The compared methods are isolated, state-only diffusion, WAC-UKF with $L=5$, the two information-diffusion means, and a centralized UKF. Besides the original single-rate case, nodes $1,3$ sample every second tick and nodes $2,4$ every third.

\begin{table}[t]
\centering
\caption{Benchmark of \cite{WeightAveCons}: steady-state network-average RMSE ($t>20$) over converged runs, with divergent runs out of 30 in brackets, and scalars broadcast per node per step.}
\label{tab:bench}
\footnotesize
\setlength{\tabcolsep}{4.5pt}
\begin{tabular}{lccc}
\hline
 & single-rate & multi-rate $3\!:\!2$ & comm.\\
\hline
isolated UKFs & 0.52 [7] & --- [30] & 0\\
one-hop state-only & 0.99 [1] & --- [30] & 2\\
WAC-UKF ($L{=}5$) \cite{WeightAveCons} & 0.70 [0] & --- [30] & 25\\
one-hop info.\ diffusion, \eqref{eq:fusemeana} & 0.99 [0] & --- [30] & 5\\
one-hop info.\ diffusion, \eqref{eq:fusemeanb} & \textbf{0.47} [0] & \textbf{0.036} [0] & 5\\
centralized & 0.63 [0] & 0.016 [1] & ---\\
\hline
\end{tabular}
\end{table}

\begin{figure}[!tb]
    \centering
    \includegraphics[width=0.95\columnwidth]{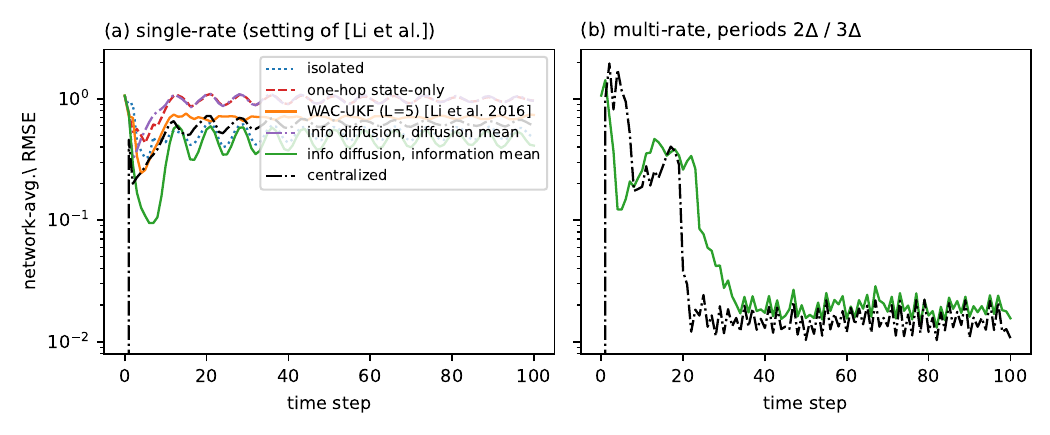}
    \caption{Benchmark of \cite{WeightAveCons}, network-average RMSE over converged runs: (a) the single-rate regime of \cite{WeightAveCons}; (b) the $3\!:\!2$ multi-rate regime, in which the covariance-averaging consensus of \cite{WeightAveCons} diverges in every run while one-hop information diffusion remains within a factor $2.3$ of the centralized filter.}
    \label{fig:bench1}
\end{figure}

\Cref{tab:bench,fig:bench1} show that the information mean is the most accurate distributed method in the single-rate case at one fifth of WAC-UKF communication. Under $3\!:\!2$ sampling, WAC-UKF, isolated, state-only, and the diffusion mean diverge in all runs, whereas the information mean never diverges and is within a factor $2.3$ of centralized. This agrees with the theory: arithmetic covariance treatment can become overconfident under sparse updates, while information averaging is conservative. The even channel $\sin(x_1)x_1$ creates a sign-symmetric basin, so the ordering is more informative than the absolute single-rate errors. The centralized UKF is not an optimal-filter bound in this nonlinear benchmark. Per-run inspection shows its single-rate deficit is carried by five outlier runs that settle -- with correct signs throughout -- on a spurious wrong-amplitude attractor near steady error $1.5$, consistent with the multiple preimages of the non-injective channel, while in the remaining runs the two are comparable (median $0.010$ vs.\ $0.012$); the outlier tail dominates the reported RMS. Its single divergent multi-rate run is an early transient in which the quadratic $x_2^{2}$ term amplifies overshoot between sparse, low-information updates near the origin; the four sign-straddling initializations and the conservative information averaging of the distributed filter avoid both failure modes; moreover, the WAC-UKF theorem assumes the uniform covariance bounds derived here~\cite[Thm.~2]{WeightAveCons}.

\subsection{Three-inertia network with individually undetectable nodes}\label{sec:simbench2}
The linear three-inertia testbed~\cite{KimLeeShim2020,LuoSuZeng2026} uses four relative-position/velocity sensors, $\Delta=0.05$ s, noise standard deviation $0.02$, and $20$ runs of $60$ s. Local observability ranks are $4,3,4,5$ of $6$; the collective rank is $5$ because the rigid-body mode $v=[1,0,1,0,1,0]^{\T}$ is invisible. Thus \cref{def:collective} holds on the quotient -- errors and covariances referenced modulo the invisible mode -- with the same $2/3$ sampling schedule.

\begin{figure}[!tb]
    \centering
    \includegraphics[width=0.95\columnwidth]{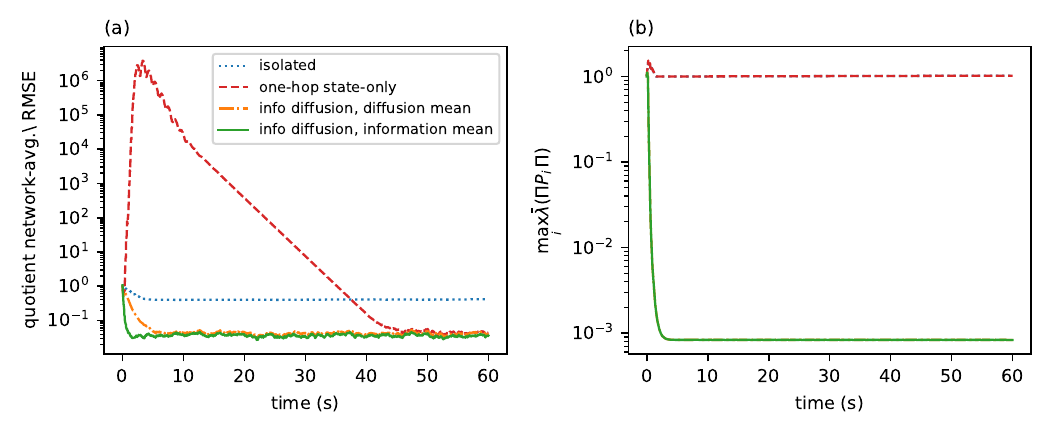}
    \caption{Three-inertia network, multi-rate sampling, quotient (rigid-body-mode-referenced) quantities: (a) network-average RMSE; (b) largest projected covariance eigenvalue $\max_{i}\bar\lambda(\Pi P_{i}\Pi)$. Isolated nodes and state-only averaging retain their initial uncertainty on the locally invisible directions, while information diffusion contracts the quotient covariance by three orders of magnitude, as \cref{thm:Pbounds3} guarantees.}
    \label{fig:bench2}
\end{figure}

In \cref{fig:bench2}, quotient RMSE is $0.41$ for isolated filters and $0.67$ for state-only averaging, but $0.042/0.036$ for the two information-diffusion means. With the neutral pair structurally invisible and dynamically decoupled, these quotient errors are covered by \cref{cor:quoterr}. The projected covariance contracts from $1$ to $10^{-3}$. Direct evaluation of \eqref{eq:gramian} with $N=8$ gives a network-wide projected eigenvalue floor $1.3\times10^{-4}$, attained at node 4. Hence the covariance contraction is guaranteed by \cref{thm:Pbounds3}, and the quotient error behaviour by \cref{cor:quoterr}. Together, the studies show that the fused mean, not the shared covariance rule, determines nonlinear robustness.

\section{Conclusion}\label{sec:conclusion}
A common base grid and one-hop information diffusion enable distributed UKFs to operate under heterogeneous sampling and collective rather than nodewise detectability. The analysis gives uniform quotient-covariance bounds for either fused mean. When the invisible subspace is trivial and the coherence margin holds, it further gives mean-square error bounds: without a mixing condition for the information mean, with one for the diffusion mean. The fixed-attenuation $H_{\infty}$ variant adds explicit feasibility, regularized-information, and contraction margins. The benchmarks confirm the predicted separation: state-only and arithmetic-covariance fusion can fail, whereas information-mean diffusion is accurate, communication-efficient, and stable under $3\!:\!2$ sampling. Future work includes adaptive fusion weights and time-varying graphs.

\FloatBarrier
\appendix
\setlength{\abovedisplayskip}{4pt plus 1pt minus 2pt}
\setlength{\belowdisplayskip}{4pt plus 1pt minus 2pt}
\setlength{\abovedisplayshortskip}{2pt plus 1pt}
\setlength{\belowdisplayshortskip}{2pt plus 1pt}
\setlength{\jot}{2pt}
\section{CUKF Stages}\label{app:ukf}
From $(\hat x,P)$, choose a symmetric weighted sigma set $\{\chi^i,W_i^{(m)},W_i^{(c)}\}_{i=0}^{N_\sigma-1}$ satisfying the usual mean and covariance reproduction identities. This includes both the equal-weight $2n_x$ set $\hat x\pm\{\sqrt{n_xP}\}_i$, used in several power-system implementations~\cite{SinghPalDecentralized,ZhaoMili2019TSP}, and the scaled $2n_x+1$ unscented set used in \cref{sec:sim}. Propagating the points through one Runge--Kutta sub-step and recombining them with $W_i^{(m)},W_i^{(c)}$ gives the predicted mean, the propagated spread $\mathcal P_f^{\sigma}$, and thence the predicted covariance $\mathcal P^{-}=\mathcal P_f^{\sigma}+Q$; repeating $\mathsf m$ times advances one base period. At a sampling tick, fresh predicted sigma points pass through $\mathscr H_l$ -- the map of whichever sensors sample at that tick -- giving $\hat z$, $P^{yy}$, and $P^{xy}$, followed by $K=P^{xy}(P^{yy})^{-1}$, $\hat x^+=\hat x^-+K(y-\hat z)$, and $P^+=P^--KP^{yy}K^{\T}$. These are the statistics used in \eqref{eq:Hpseudo}--\eqref{eq:Jinc}; the proofs require only the stated covariance identities, positive definiteness, and residual bounds, not a particular symmetric rule. Thus a chosen rule is covered only on trajectories for which the conditions in \cref{ass:bounds} hold.

\section{Proofs of the Theorems}\label{app:proofs}
\begingroup\footnotesize
\subsection{Proof of \texorpdfstring{\cref{lem:rank}}{Lemma 1}}
The first claim is the uniform observability notion for time-varying discrete-time pairs~\cite{anderson1981detectability}: the windowed Gramian $\mathcal O_{l,n}^{\T}\mathcal O_{l,n}=\sum_k\Phi^{\T}(k,n)J_{l,k}\Phi(k,n)$ bounded below by $\beta_lI$ uniformly in $n$ is uniform complete observability, whereas a rank condition alone allows the smallest eigenvalue to decay with $n$. For the second, let the columns of $W$ span the fixed nullspace, those of $U$ an orthonormal complement, and take $\Theta_l=[U\;W]^{\T}$ (a general uniformly conditioned $\Theta_l$ changes the constants by its condition number only). Invariance gives $U^{\T}F_{l,n}W=0$, so $\Theta_lF_{l,n}\Theta_l^{-1}$ is block-triangular. Every $\mathbf x\in\ker\mathcal O_{l,n}$ satisfies $J_{l,k}^{1/2}\Phi(k,n)\mathbf x=0$ blockwise, hence $\mathcal O_{l,n}\Theta_l^{-1}=[\,\mathcal O^{ob}_{l,n}\;\;0\,]$ and the reduced pair's windowed Gramian is the observable block $U^{\T}\mathcal O_{l,n}^{\T}\mathcal O_{l,n}U$; a uniform lower bound on it is uniform complete observability of the reduced pair. $\hfill\blacksquare$

\subsection{Proof of \texorpdfstring{\cref{stablemma}}{Theorem 1}}
Drop the node index and let $\varrho_n\in\{0,1\}$ indicate whether the local measurement is present. Put $\widetilde Q_{n-1}=Q_{n-1}+\Sigma^f_{n-1}$ and $\widetilde R_n=R_n+\Sigma_n$, with $\Sigma^{f},\Sigma$ the regression residuals of \cref{sec:linearization}. Define the \emph{computed} closed-loop matrix
\[
B_n\triangleq(I-\varrho_n\mathcal K_nH_n)F_{n-1}.
\]
The standard Joseph identity~\cite{simon2006optimal}, applied to \eqref{eq:Ppred}--\eqref{eq:Kgain}, gives the exact computed-covariance recursion
\begin{equation}\label{eq:Precur}
\mathcal P_n=B_n\mathcal P_{n-1}B_n^{\T}+W_n,
\end{equation}
where
\begin{equation*}
W_n=(I-\varrho_n\mathcal K_nH_n)\widetilde Q_{n-1}
(I-\varrho_n\mathcal K_nH_n)^{\T}+\varrho_n\mathcal K_n\widetilde R_n\mathcal K_n^{\T}\succeq0.
\end{equation*}
This recursion contains only quantities computed by the filter.

By the covariance bounds and \cref{ass:bounds}, $\mathcal P_{n|n-1}\preceq\bar p^-I$ and $\|\mathcal K_n\|\leq\bar K$. At a measurement tick,
$I-\mathcal K_nH_n=\mathcal P_n\mathcal P_{n|n-1}^{-1}$, hence
$\sigma_{\min}(I-\mathcal K_nH_n)\geq(\bar p\bar p^-)^{-1}$; at a silent tick this matrix is $I$. Thus $W_n\succeq\sigma_0^2\underline q I$ and $\|B_n\|\leq\bar f(1+\bar K\bar h)$. Since both $F_{n-1}$ and $I-\varrho_n\mathcal K_nH_n$ are nonsingular, \eqref{eq:Precur} implies
\begin{equation}\label{eq:contraction}
B_n^{\T}\mathcal P_n^{-1}B_n\preceq r_0\mathcal P_{n-1}^{-1}.
\end{equation}
Indeed, factor $B_n$ from \eqref{eq:Precur}, use
$B_n^{-1}W_nB_n^{-\T}\succeq\delta I\succeq\delta\underline p\mathcal P_{n-1}$, and invert.

The true error recursion follows from \eqref{eq:errdyn}--\eqref{eq:errobs}:
\[
\widetilde{\mathbf x}_n=B_n\widetilde{\mathbf x}_{n-1}+T_n\widetilde{\mathbf x}_{n-1}+\eta_n,
\]
with
\begin{align*}
T_n={}&(I-\varrho_n\mathcal K_nH_n)E_{n-1}
-\varrho_n\mathcal K_nD_n(F_{n-1}+E_{n-1}), &\|T_n\|&\leq\bar d_*.
\end{align*}
and disturbance term
$\eta_n=(I-\varrho_n\mathcal K_n(H_n+D_n))\mathbf w_{n-1}-\varrho_n\mathcal K_n\mathbf v_n$.
For $V_n=\widetilde{\mathbf x}_n^{\T}\mathcal P_n^{-1}\widetilde{\mathbf x}_n$, \eqref{eq:contraction} gives a baseline contribution at most $r_0V_{n-1}$, while
\[
(T_n\widetilde{\mathbf x})^{\T}\mathcal P_n^{-1}(T_n\widetilde{\mathbf x})
\leq(\bar p/\underline p)\bar d_*^2V_{n-1}.
\]
Cauchy--Schwarz in the $\mathcal P_n^{-1}$ inner product therefore bounds the complete noise-free map by
$\bigl(\sqrt{r_0}+\sqrt{\bar p/\underline p}\,\bar d_*\bigr)^2V_{n-1}\leq r_{\mathrm{loc}}V_{n-1}$ under \eqref{eq:localdiscrepancy}. The coefficients in $\eta_n$ need not be measurable before the contemporaneous noises. Nevertheless, their norms are uniformly bounded, and the white-noise assumptions give a finite constant $\bar\mu_v$ such that
$\mathbb E[\eta_n^{\T}\mathcal P_n^{-1}\eta_n\mid\mathcal F_{n-1}]\leq\bar\mu_v$.
Applying Young's inequality in the $\mathcal P_n^{-1}$ metric with $\tau_{\mathrm{loc}}=(1-r_{\mathrm{loc}})/(2r_{\mathrm{loc}})$, rather than discarding a cross term, gives
\[
\begin{aligned}
\mathbb E[V_n\mid\mathcal F_{n-1}]
&\leq \rho_{\mathrm{loc}}V_{n-1}
 +(1+\tau_{\mathrm{loc}}^{-1})\bar\mu_v,\\
\rho_{\mathrm{loc}}&=\frac{1+r_{\mathrm{loc}}}{2}<1.
\end{aligned}
\]
Iteration and $\underline p\|\widetilde{\mathbf x}_n\|^2\leq V_n$ give exponential mean-square boundedness. \hfill$\blacksquare$

\subsection{Auxiliary lemmas}

\begin{lemma}[adapted from {\cite[Lemma~5]{wei2018}}]\label{lem:alpha}
Let \cref{ass:bounds}.1 hold and suppose $\mathcal P_{l,n}^{-1}\preceq\bar pI$. With
\[
\alpha\triangleq(1+\bar p\bar q\check f^2)^{-1},
\]
the computed prediction satisfies
\[
\mathcal P_{l,n+1|n}^{-1}\succeq
\alpha F_{l,n}^{-\T}\mathcal P_{l,n}^{-1}F_{l,n}^{-1}.
\]
\end{lemma}
\begin{proof}
Write $\widetilde Q_n=Q_n+\Sigma^f_{l,n}$. From $\|F^{-1}\|\leq\check f$ and $\mathcal P^{-1}\preceq\bar pI$,
$\widetilde Q_n\preceq\bar qI\preceq\bar q\check f^2FF^{\T}\preceq\bar p\bar q\check f^2F\mathcal PF^{\T}$. Hence \eqref{eq:Ppred} gives
$\mathcal P_{n+1|n}\preceq(1+\bar p\bar q\check f^2)F\mathcal PF^{\T}$; inversion proves the claim.
\end{proof}

\begin{lemma}[exact mixing for the information mean]\label{lem:mix}
Let $\Omega_{j}\succ0$, $c_{j}\geq0$ with $\sum_{j}c_{j}=1$, and set $\Omega\triangleq\sum_{j}c_{j}\Omega_{j}$ and $\mathbf{e}\triangleq\Omega^{-1}\sum_{j}c_{j}\Omega_{j}\mathbf{e}_{j}$. Then
$\mathbf{e}^{\T}\Omega\,\mathbf{e}\leq\sum_{j}c_{j}\,\mathbf{e}_{j}^{\T}\Omega_{j}\mathbf{e}_{j}$.
\end{lemma}
\begin{proof}
Each matrix $\bigl[\begin{smallmatrix}I\\ \mathbf{e}_{j}^{\T}\end{smallmatrix}\bigr]\Omega_{j}\bigl[\begin{smallmatrix}I&\mathbf{e}_{j}\end{smallmatrix}\bigr]\succeq0$, hence so is their convex combination $\bigl[\begin{smallmatrix}\Omega&\mathbf{u}\\ \mathbf{u}^{\T}&s\end{smallmatrix}\bigr]$ with $\mathbf{u}=\sum_{j}c_{j}\Omega_{j}\mathbf{e}_{j}$ and $s=\sum_{j}c_{j}\mathbf{e}_{j}^{\T}\Omega_{j}\mathbf{e}_{j}$. Its Schur complement with respect to $\Omega\succ0$ gives $s-\mathbf{u}^{\T}\Omega^{-1}\mathbf{u}\geq0$, which is the claim since $\mathbf{e}^{\T}\Omega\mathbf{e}=\mathbf{u}^{\T}\Omega^{-1}\mathbf{u}$.
\end{proof}

\subsection{Proof of \texorpdfstring{\cref{thm:Pbounds}}{Theorem 2}}

Consider a node $l$ whose measurement update at each base tick assimilates the neighbourhood information, so that
\begin{equation}\label{eq:inforecur}
\mathcal P_{l,n}^{-1}=\mathcal P_{l,n|n-1}^{-1}+S_{l,n},
\end{equation}
where $S_{l,n}\triangleq\textstyle\sum_{j\in\mathcal{N}_{l}(n)}J_{j,n}$ collects the realized information increments \eqref{eq:Jinc} of the neighbours sampling at tick $n$ ($\mathcal{N}_{l}(n)\subseteq\mathcal{N}_{l}$).

\emph{Upper bound.} From \eqref{eq:Ppred}, $\mathcal P_{l,n|n-1}\succeq Q_{n-1}\succeq\underline q\,I$, hence $\mathcal P_{l,n|n-1}^{-1}\preceq\underline q^{-1}I$; and $S_{l,n}\preceq|\mathcal{N}_{l}|\,\bar h^{2}\underline r^{-1}I$. Therefore
$$
\mathcal P_{l,n}^{-1}\preceq\Big(\frac{1}{\underline q}+\frac{\bar N\,\bar h^{2}}{\underline r}\Big)I\triangleq\bar p\,I,\qquad \bar N\triangleq\max_{l}|\mathcal{N}_{l}|.
$$

\emph{Lower bound.} The upper bound legitimizes \cref{lem:alpha}: with $\alpha=(1+\bar p\,\bar q\,\check f^{2})^{-1}$,
$\mathcal P_{l,n|n-1}^{-1}\succeq\alpha\,\Phi_{l}(n,n{-}1)^{-\T}\mathcal P_{l,n-1}^{-1}\Phi_{l}(n,n{-}1)^{-1}$, with $\Phi_{l}$ the node's own realized transition products of \cref{def:collective}. Substituting into \eqref{eq:inforecur} and iterating over $k=n-N+1,\dots,n$, dropping the positive terminal term,
\begin{align*}
\mathcal P_{l,n}^{-1}
&\succeq\sum_{k=n-N+1}^{n}\alpha^{\,n-k}
\Phi_l(k,n)^{\T}S_{l,k}\Phi_l(k,n)\succeq\alpha^{N-1}\beta\Pi.
\end{align*}
by \eqref{eq:gramian}, since $S_{l,k}=J^{\mathrm{col}}_{l,k}$ is exactly the collective realized increment of \cref{def:collective}. Hence $\underline p\triangleq\alpha^{N-1}\beta$ for $n\geq N$; for $n<N$, the initialization $\mathcal P_{l,0}^{-1}\succ0$ and the same one-step inequalities give a positive initialization-dependent bound. Only node $l$'s own realized products appear, so no detectability of any scaled pair, and no common linearization across nodes, is invoked.

\smallskip
\noindent\emph{Scope.} The lower bound hinges on the neighbourhood term $S_{l,n}$ in \eqref{eq:inforecur}. In \cref{DCUKF1} each node assimilates only its \emph{own} measurement, $S_{l,n}=J_{l,n}\,\mathbf{1}\{\text{node }l\text{ samples at }n\}$, and the argument then requires the \emph{local} pair to be uniformly detectable at node $l$'s own sampling schedule; collective detectability alone is not sufficient (cf.\ \cref{rem:alg1caveat} and \cref{sec:sim}).

\subsection{Proof of \texorpdfstring{\cref{prop:alg2}}{Proposition 1}}
On its locally observable coordinates, \cref{DCUKF2} runs the single-node information recursion with the node's own channel: $(\mathcal P^{ob}_{l,n})^{-1}=(\mathcal P^{ob}_{l,n|n-1})^{-1}+S^{ob}_{l,n}$, with $S^{ob}_{l,n}$ the reduced own-channel information increment at the ticks where node $l$ samples. The upper bound follows as in \cref{thm:Pbounds}. The reduced realized dynamics inherit the bounds of \cref{ass:bounds} through the uniformly conditioned $\Theta_{l}$, so \cref{lem:alpha} applies on the reduced block, and the unroll of the proof of \cref{thm:Pbounds} -- with $S_{l,k}$ replaced by $S^{ob}_{l,k}$ and \eqref{eq:gramian} assumed on the reduced coordinates with $\mathcal V=\{0\}$ (uniform regularity) -- yields $\underline p_{2}\triangleq\alpha_{2}^{N_{2}-1}\beta_{2}$, with constants depending additionally on the conditioning of $\Theta_{l}$. $\hfill\blacksquare$

\subsection{Proofs of \texorpdfstring{\cref{cor:alg2err} and \cref{cor:quoterr}}{Corollaries 1 and 2}}
\emph{Proof of \cref{cor:alg2err}.}\ Under the regularity hypothesis the reduced computed recursion is a self-contained CUKF instance on the observable block of $\Theta_lF\Theta_l^{-1}$ with increments $J^{ob}$, and \cref{prop:alg2} supplies $(\underline p_2,\bar p_2)$. The reduced exact error relation carries the $\Theta_l$-restricted discrepancies (plus, for $E\neq0$, the bounded complement coupling, absorbed into $\mu$), so the Joseph/Young argument of \cref{stablemma} applies verbatim with the reduced constants, giving $\mathbb E[V^{ob}_n\mid\mathcal F_{n-1}]\leq\rho^{(2)}_{\mathrm{loc}}V^{ob}_{n-1}+(1+\tau_2^{-1})\bar\mu_2$. In $l$'s observable coordinates, Step~5 is a convex combination of neighbour contributions; under the compatibility condition each contribution is itself a filtered quantity obeying the same recursion, and the norm-equivalence factor $\kappa_2$ yields the network bound as in the diffusion branch of \cref{thm:errbound}. Without compatibility the combination includes unfiltered components -- the injection mechanism of \cref{sec:sim}. $\hfill\blacksquare$

\smallskip\noindent\emph{Proof of \cref{cor:quoterr}.}\ Exact invisibility gives $U^{\T}F_n=F^{q}_nU^{\T}$, $H_n=H^{q}_nU^{\T}$ ($F^{q}_n\triangleq U^{\T}F_nU$, $H^{q}_n\triangleq H_nU$), whence $U^{\T}\mathcal P_{n|n-1}U$, the innovation covariance, $U^{\T}\mathcal K$, and $U^{\T}\mathcal P_nU$ obey the quotient filter's prediction, innovation, gain, and update recursions, for any cross-covariance -- (i). Since $U^{\T}(I-\mathcal KH)F=(I-U^{\T}\mathcal K H^{q})F^{q}U^{\T}$, the quotient error recursion is autonomous; scalar diffusion weights commute with $U^{\T}$, so \cref{thm:errbound}'s diffusion branch applies with the constants of (i) -- (ii). Under $W^{\T}F_nU=0$ the recursions preserve block-diagonality of $\mathcal P$, hence $M_j=\mathcal P_{\mathrm{new}}\mathcal P_j^{-1}$ is block-diagonal and the information-mean combination commutes with $U^{\T}$ -- (iii). $\hfill\blacksquare$

\subsection{Proof of \texorpdfstring{\cref{thm:Pbounds3}}{Theorem 3}}

By \cref{DCUKF3}, $\mathcal P_{l,n}^{-1}=\sum_{j\in\mathcal{N}_{l}}c_{j,l}\big[\mathcal P_{j,n|n-1}^{-1}+S_{j,n}\big]$ with $S_{j,n}\triangleq J_{j,n}\cdot\mathbf{1}\{j\text{ samples at }n\}$ the realized information increments \eqref{eq:Jinc}.

\emph{Upper bound.} Each bracket satisfies $\mathcal P_{j,n|n-1}^{-1}+S_{j,n}\preceq(\underline q^{-1}+\bar h^{2}\underline r^{-1})I$ by the argument of \cref{thm:Pbounds}; a convex combination obeys the same bound, so $\bar p\triangleq\underline q^{-1}+\bar h^{2}\underline r^{-1}$ works (the neighbourhood factor $\bar N$ is not needed here because each node assimilates only its own measurement before diffusing).

\emph{Lower bound.} With $\alpha_3=(1+\bar p_3\bar q\check f^2)^{-1}$, keeping only the self term of the prediction part and every measurement term ($c_{j,l}\geq\underline c$),
$$
\mathcal P_{l,n}^{-1}\;\succeq\;\underline c\,\mathcal P_{l,n|n-1}^{-1}+\underline c\,S^{\mathrm{col}}_{l,n},\qquad S^{\mathrm{col}}_{l,n}\triangleq\!\!\sum_{j\in\mathcal{N}_{l}(n)}\!\!S_{j,n},
$$
and \cref{lem:alpha} on node $l$'s own chain gives $\mathcal P_{l,n|n-1}^{-1}\succeq\alpha_3\,\Phi_{l}(n,n{-}1)^{-\T}\mathcal P_{l,n-1}^{-1}\Phi_{l}(n,n{-}1)^{-1}$ with the computed $F$-products. Iterating over the window of \cref{def:collective},
\begin{align*}
\mathcal P_{l,n}^{-1}&\succeq\underline c\!\!\sum_{k=n-N+1}^{n}\!\!(\alpha_3\underline c)^{\,n-k}\,\Phi_{l}(k,n)^{\T}S^{\mathrm{col}}_{l,k}\,\Phi_{l}(k,n)\succeq\underline c\,(\alpha_3\underline c)^{N-1}\beta\,\Pi,
\end{align*}
so $\underline p_{3}\triangleq\underline c\,(\alpha_3\underline c)^{N-1}\beta$ for $n\geq N$. Only node $l$'s own realized transition products appear -- no common linearization across nodes, no Loewner minimum, and no scaled-pair detectability are invoked -- and the covariance recursion, hence the bound, is identical for either mean rule in Step 3 of \cref{DCUKF3}. Finally, for the quotient consequence: from $\mathcal P^{-1}\succeq\underline p_{3}\Pi$, taking $\mathbf{x}=\mathcal P U\mathbf{z}$ in $\mathbf{x}^{\T}\mathcal P^{-1}\mathbf{x}\geq\underline p_{3}\|U^{\T}\mathbf{x}\|^{2}$ gives $\mathbf{z}^{\T}(U^{\T}\mathcal PU)\mathbf{z}\geq\underline p_{3}\,\mathbf{z}^{\T}(U^{\T}\mathcal PU)^{2}\mathbf{z}$, i.e., $U^{\T}\mathcal PU\succeq\underline p_{3}(U^{\T}\mathcal PU)^{2}$, hence $U^{\T}\mathcal PU\preceq\underline p_{3}^{-1}I$. $\hfill\blacksquare$

\subsection{Proof of \texorpdfstring{\cref{thm:errbound}}{Theorem 4}}
Use superscripts $F$ and $L$ for the fused covariance and the local posterior before fusion. The prediction at tick $n+1$ starts from $\mathcal P^F_{j,n}$. For each node, the exact computed recursion from $\mathcal P^F_{j,n}$ to $\mathcal P^L_{j,n+1}$ is the Joseph recursion \eqref{eq:Precur}, with the stage-uniform bounds $(\underline p,\bar p)$ stated in the theorem. Therefore the local argument of \cref{stablemma}, with those constants, gives
\begin{equation}\label{eq:localdec}
\mathbb E\!\left[(\widetilde{\boldsymbol\psi}_{j,n+1})^{\T}
(\mathcal P^L_{j,n+1})^{-1}\widetilde{\boldsymbol\psi}_{j,n+1}
\mid\mathcal F_n\right]\leq \rho_{\mathrm{loc}}(\widetilde{\mathbf x}_{j,n})^{\T}
(\mathcal P^F_{j,n})^{-1}\widetilde{\mathbf x}_{j,n}+\mu_0.
\end{equation}
where $\mu_0<\infty$ is uniform. The only nonlinear perturbation is
$(I-\mathcal KH)E-\mathcal KD(F+E)$, whose norm is bounded by $\bar d_*$; hence \eqref{eq:localdiscrepancy} is exactly the condition used in obtaining \eqref{eq:localdec}.

For the information mean, \cref{lem:mix} with $\Omega_j=(\mathcal P^L_{j,n+1})^{-1}$ gives
\begin{equation*}
(\widetilde{\mathbf x}_{l,n+1})^{\T}
(\mathcal P^F_{l,n+1})^{-1}\widetilde{\mathbf x}_{l,n+1} \leq\sum_{j\in\mathcal N_l}c_{j,l}
(\widetilde{\boldsymbol\psi}_{j,n+1})^{\T}
(\mathcal P^L_{j,n+1})^{-1}\widetilde{\boldsymbol\psi}_{j,n+1}.
\end{equation*}
Taking expectations first and then maximizing over nodes yields $W_{n+1}\leq \rho_{\mathrm{loc}}W_n+\mu_0$.

For the diffusion mean, convexity gives the same inequality with the additional metric-equivalence factor
$\kappa=\bar p/\underline p$, because $(\mathcal P^F_{l,n+1})^{-1}\preceq\bar pI\preceq\kappa(\mathcal P^L_{j,n+1})^{-1}$. Thus
$W_{n+1}\leq\kappa \rho_{\mathrm{loc}}W_n+\kappa\mu_0$, which contracts precisely under \eqref{eq:gaincondition}. In either case, the lower information bound converts the weighted estimate into the asserted Euclidean mean-square bound. \hfill$\blacksquare$

\subsection{Proof of \texorpdfstring{\cref{thm:hinfP}}{Theorem 5}}
Write $\Omega_{l,n}=\mathcal P_{l,n|n-1}^{-1}$ and
\begin{align*}
G_{l,n}&=\mathcal I_{l,n}
=\sum_{j\in\mathcal N_l(n)}H_{j,n}^{\T}R_{j,n}^{-1}H_{j,n},\\
M_{l,n}&=\Omega_{l,n}+G_{l,n}.
\end{align*}
These are exactly the quantities computed by \cref{UKF22}. Since $\Sigma_{j,n}\succeq0$,
$H_j^{\T}R_j^{-1}H_j\succeq J_{j,n}$, so $G_{l,n}\succeq J^{\mathrm{col}}_{l,n}$.

Assume inductively that $M_{l,n-1}\succeq\underline p_0I$. By \cref{ass:hinf},
\[
\mathcal P_{l,n-1}^{-1}=M_{l,n-1}-\theta^{-2}I\succeq\tfrac12M_{l,n-1}\succ0.
\]
The prediction is therefore well posed. Moreover, $\Omega_{l,n}\preceq\underline q^{-1}I$ and $G_{l,n}\preceq\bar d\bar h^2\underline r^{-1}I$, which gives the upper bound $M_{l,n}\preceq\bar p_HI$ (including the finite initial value by the definition of $\bar p_H$).

For the lower bound, \cref{lem:alpha} and the preceding shrink inequality give
\[
M_{l,n}\succeq\frac{\alpha_H}{2}F_{l,n-1}^{-\T}M_{l,n-1}F_{l,n-1}^{-1}+J^{\mathrm{col}}_{l,n}.
\]
Unrolling this inequality over the $N$-tick window and applying \eqref{eq:gramian} with $\Pi=I$ yields
$M_{l,n}\succeq(\alpha_H/2)^{N-1}\beta I=\underline p_0I$ for $n\geq N$; the finite-prefix requirement in \cref{ass:hinf} covers $0\leq n<N$. To verify the sufficient initialization stated there, let $m_n=\lambda_{\min}(M_{l,n})$ and use \cref{lem:alpha} before invoking the established lower margin:
\[
m_n\geq g\bigl(m_{n-1}-\theta^{-2}\bigr),\qquad g=\alpha_H/\bar f^2.
\]
For $0<g<1$, iteration gives
\[
m_n\geq g^n m_0-\theta^{-2}\frac{g(1-g^n)}{1-g}.
\]
The displayed initialization in \cref{ass:hinf} is a conservative bound that makes the right-hand side at least $\underline p_0$ for every $0\leq n<N$. Finally
$\mathcal P_{l,n}^{-1}=M_{l,n}-\theta^{-2}I\succeq(\underline p_0/2)I$, completing the proof. \hfill$\blacksquare$

\subsection{Proof of \texorpdfstring{\cref{thm:hinferr}}{Theorem 6}}
A covariance-form route through the cross term is unavailable: $R_{e}$ is indefinite by construction (a scalar case with $\mathcal P{=}0.5,R{=}2,S{=}1,\theta^{2}{=}1$ gives cross term $-2/9$ despite a valid regularized-information margin). The proof therefore works in the information metric, where the $\theta^{-2}$ subtraction has a definite sign effect.
Fix a node and suppress its index. Put $\Omega_n=\mathcal P_{n|n-1}^{-1}$, $G_n=\mathcal I_n$, $M_n=\Omega_n+G_n$, and $W_n=\mathcal P_n^{-1}=M_n-\theta^{-2}I$. By \cref{thm:hinfP},
$(\underline p_0/2)I\preceq W_n\preceq\bar p_HI$ and $M_n\preceq\bar p_HI$.

The exact error relations and \eqref{eq:crossdelta} give
\begin{align*}
\widetilde{\boldsymbol\psi}_n={}&
M_n^{-1}(\Omega_n-\Upsilon_n)(F_{n-1}+E_{n-1})
\widetilde{\mathbf x}_{n-1}+M_n^{-1}(\Omega_n-\Upsilon_n)\mathbf w_{n-1}-M_n^{-1}\boldsymbol\xi_n.
\end{align*}
where $\boldsymbol\xi_n=\sum_jH_j^{\T}R_j^{-1}\mathbf v_j$ is the aggregated measurement-noise term. Split the noise-free map into the computed baseline
$A_{0,n}=M_n^{-1}\Omega_nF_{n-1}$ and the perturbation
\[
T_{\theta,n}=M_n^{-1}\!\left[\Omega_nE_{n-1}-\Upsilon_n(F_{n-1}+E_{n-1})\right].
\]

First, \eqref{eq:Ppred} and the prior stored-information lower bound imply
\[
F_{n-1}^{\T}\Omega_nF_{n-1}\preceq(1+\delta_\theta)^{-1}\mathcal P_{n-1}^{-1}.
\]
Second,
\begin{align*}
\Omega_nM_n^{-1}W_nM_n^{-1}\Omega_n
={}&\Omega_nM_n^{-1}\Omega_n-\theta^{-2}\Omega_nM_n^{-2}\Omega_n\preceq{}(1-\theta^{-2}\gamma_\theta)\Omega_n.
\end{align*}
Indeed, \cref{lem:alpha} and the prior stored-information lower bound give $\lambda_{\min}(\Omega_n)\geq\underline\omega$, while $\lambda_{\max}(M_n)\leq\bar p_H$. Consequently
\[
A_{0,n}^{\T}W_nA_{0,n}\preceq r_{\theta,0}\mathcal P_{n-1}^{-1}.
\]
Moreover, $\|M_n^{-1}\|\leq\underline p_0^{-1}$, $\|\Omega_n\|\leq\underline q^{-1}$, and hence $\|T_{\theta,n}\|\leq\bar d_\theta$. Therefore
\begin{equation*}
(T_{\theta,n}\mathbf x)^{\T}W_n(T_{\theta,n}\mathbf x) \leq\frac{2\bar p_H}{\underline p_0}\bar d_\theta^2\,
\mathbf x^{\T}\mathcal P_{n-1}^{-1}\mathbf x.
\end{equation*}
Cauchy--Schwarz in the $W_n$ metric and \eqref{eq:hinfdiscrepancy} show that the full noise-free map contracts by at most $r_\theta$.

Let
\[
\eta_{\theta,n}=M_n^{-1}(\Omega_n-\Upsilon_n)\mathbf w_{n-1}-M_n^{-1}\boldsymbol\xi_n.
\]
Its coefficients may depend on the contemporaneous nonlinear trajectory, so no conditional-centering argument is used. The bounds $\|M_n^{-1}(\Omega_n-\Upsilon_n)\|\leq\underline p_0^{-1}(\underline q^{-1}+\bar d_\Upsilon)$ and $W_n\preceq\bar p_HI$, together with the uniformly bounded noise covariances, give a finite constant $\bar\mu_\theta$ such that
$\mathbb E[\eta_{\theta,n}^{\T}W_n\eta_{\theta,n}\mid\mathcal F_{n-1}]\leq\bar\mu_\theta$.
Applying Young's inequality in the $W_n$ metric with $\tau_\theta=(1-r_\theta)/(2r_\theta)$ yields
\begin{equation*}
\mathbb E\!\left[(\widetilde{\boldsymbol\psi}_n)^{\T}W_n
\widetilde{\boldsymbol\psi}_n\mid\mathcal F_{n-1}\right] \leq \rho_\theta(\widetilde{\mathbf x}_{n-1})^{\T}
\mathcal P_{n-1}^{-1}\widetilde{\mathbf x}_{n-1}
+(1+\tau_\theta^{-1})\bar\mu_\theta,
\end{equation*}
where $\rho_\theta=(1+r_\theta)/2<1$. The diffusion mean introduces the factor $\kappa_\theta=2\bar p_H/\underline p_0$, giving a contracting network recursion under \eqref{eq:hinfgain}. The lower bound on the stored information then converts the weighted bound to the asserted Euclidean mean-square bound. \hfill$\blacksquare$

\endgroup
\renewcommand{\bibfont}{\scriptsize}
\setlength{\bibsep}{1pt}
\bibliographystyle{elsarticle-harv}
\bibliography{ref_theory_fixed}
\end{document}